\documentclass[sigconf]{acmart}

\newif\ifcamera

\usepackage{latexsym,amsmath,amsfonts,stmaryrd,mathtools}
\usepackage{amsthm}
\usepackage{algorithm, algorithmicx}
\usepackage[noend]{algpseudocode}
\usepackage{graphicx}
\usepackage{hyperref}
\usepackage{wrapfig}
\usepackage{listings, fancyvrb, multicol}
\usepackage{multirow}
\usepackage{times}
\usepackage{comment}
\usepackage{xspace}
\usepackage{pifont}
\usepackage{parcolumns}
\usepackage{graphicx}
\usepackage{enumerate}
\usepackage{enumitem}
\usepackage{wrapfig}

\newtheorem{theorem}{Theorem}[section]
\newtheorem{lemma}[theorem]{Lemma}

\newtheorem{definition}{Definition}
\newtheorem{observation}{Observation}

\def\t{\textit}
\def\rdma{\textsc{rdma}\xspace}
\def\mr{\t{mr}}  
\def\ack{\t{ack}}
\def\nak{\t{nak}}
\def\RW{\t{RW}}
\newcommand\rmv[1]{}

\DeclareRobustCommand*\cal{\mathcal}

\newcommand{\disk}{memory}
\newcommand{\disks}{memories}
\newcommand{\mregion}{memory region}
\newcommand{\dpermission}{ dynamic permission}

\newcommand{\fastbyz}{Cheap Quorum}
\newcommand{\robustbyz}{Robust Backup}
\newcommand{\crashAlgo}{Protected Memory Paxos}
\newcommand{\combinedAlgo}{Aligned Paxos}

\newcommand{\prefPax}{Preferential Paxos}
\newcommand{\composedAlgo}{Fast \& Robust}

\newcommand{\xparagraph}[1]{\noindent\textbf{#1}}

\makeatletter
\lst@Key{countblanklines}{true}[t]%
{\lstKV@SetIf{#1}\lst@ifcountblanklines}

\lst@AddToHook{OnEmptyLine}{%
	\lst@ifnumberblanklines\else%
	\lst@ifcountblanklines\else%
	\advance\c@lstnumber-\@ne\relax%
	\fi%
	\fi}
\makeatother

\begin{document}
	
\ifcamera
\settopmatter{printacmref=true}
\else
\settopmatter{printacmref=false}
\setcopyright{none}
\renewcommand\footnotetextcopyrightpermission[1]{}
\fi
\fancyhead{}

	\ifcamera
	\title{The Impact of RDMA on Agreement}
	\else
    \title{\textbf{The Impact of RDMA on Agreement}\\ {\normalfont [Extended Version]}}
    \fi
	\author{Marcos K. Aguilera}
	\affiliation{VMware}
	\email{maguilera@vmware.com}
	
	\author{Naama Ben-David}
	\affiliation{CMU}
	\email{nbendavi@cs.cmu.edu}
	
	\author{Rachid Guerraoui}
	\affiliation{EPFL}
	\email{rachid.guerraoui@epfl.ch}
	
	\author{Virendra Marathe}
	\affiliation{Oracle}
	\email{virendra.marathe@oracle.com}
	
	\author{Igor Zablotchi}
	\affiliation{EPFL}
	\email{igor.zablotchi@epfl.ch}
	
	\renewcommand{\shortauthors}{Aguilera, Ben-David, Guerraoui, Marathe and Zablotchi}
	
	\begin{abstract}
Remote Direct Memory Access (RDMA) is becoming widely available in data centers.
This technology allows a process to directly read and write the memory of a remote host,
  with a mechanism to control access permissions.
In this paper, we study the fundamental power of these capabilities.
We consider the well-known problem of achieving consensus despite failures, and
  find that RDMA can improve the inherent trade-off in distributed computing
  between failure resilience and performance.
Specifically, we show that RDMA allows algorithms that simultaneously
  achieve high resilience and high performance, while traditional algorithms had to choose
  one or another.
With Byzantine failures, we give an algorithm that only requires
  $n \geq 2f_P + 1$ processes (where $f_P$ is the maximum number of faulty processes) and decides in two (network) delays
  in common executions.
With crash failures, we give an algorithm that only requires
  $n \geq f_P + 1$ processes and also decides in two delays.
Both algorithms tolerate a minority of memory failures inherent to RDMA, and they
  provide safety in asynchronous systems and liveness with standard
  additional assumptions.
\end{abstract}

\ifcamera	
\copyrightyear{2019} 
\acmYear{2019} 
\setcopyright{acmlicensed}
\acmConference[PODC '19]{2019 ACM Symposium on Principles of Distributed Computing}{July 29-August 2, 2019}{Toronto, ON, Canada}
\acmBooktitle{2019 ACM Symposium on Principles of Distributed Computing (PODC '19), July 29-August 2, 2019, Toronto, ON, Canada}
\acmPrice{15.00}
\acmDOI{10.1145/3293611.3331601}
\acmISBN{978-1-4503-6217-7/19/07}
\fi

	\maketitle

	\section{Introduction}

In recent years, a technology known as Remote Direct Memory Access (RDMA) has made its way into data centers, earning a spotlight in distributed systems research.
RDMA provides the traditional send/receive communication primitives, but also
  allows a process to directly read/write remote memory.
Research work shows that RDMA leads to some new and exciting distributed  
  algorithms~\cite{behrens2016derecho,poke2015dare,wang2017apus,dragojevic2014farm,kalia2015using,aguilera2018passing}.
  
RDMA provides a different interface from previous communication mechanisms, as
  it combines message-passing with shared-memory~\cite{aguilera2018passing}.
Furthermore, to safeguard the remote memory, RDMA provides \textit{protection} mechanisms to
  grant and revoke access for reading and writing data.
This mechanism is fine grained: an application can choose subsets of remote memory called \emph{regions} 
  to protect;
  it can choose whether a region can be read, written, or both; 
  and it can choose individual processes to be given access, where different processes can have different
  accesses.
Furthermore, protections are \emph{dynamic}: they can be changed by the application over time.
In this paper, we lay the groundwork for a theoretical understanding of these RDMA capabilities, and
  we show that they lead to distributed algorithms that are inherently 
  more powerful than before.

While RDMA brings additional power, it also introduces some challenges.
With RDMA, the remote memories are subject to failures that cause them to become unresponsive.
This behavior differs from traditional shared memory, which is often assumed to be reliable\footnote{There
  are a few studies of failure-prone memory, as we discuss in related work.}.
In this paper, we show that the additional power of RDMA more than compensates for these challenges.

Our main contribution is to show that RDMA improves on the fundamental trade-off in distributed
  systems between failure resilience and performance---specifically, we show 
  how a consensus protocol can use RDMA to
  achieve \emph{both} high resilience and high performance, while traditional algorithms
  had to choose one or another.
We illustrate this on the fundamental problem of achieving consensus and
  capture the above RDMA capabilities as an M\&M model~\cite{aguilera2018passing}, in which
  processes can use both message-passing and shared-memory.
We consider asynchronous systems and require safety in all
  executions and liveness under standard additional assumptions
  (e.g., partial synchrony).
We measure resiliency by the number of failures an algorithm tolerates, and
  performance by the number of (network) delays in common-case
  executions.
Failure resilience and performance depend on whether processes fail by crashing or by
  being Byzantine, so we consider both.
  
With Byzantine failures, we consider the consensus problem called
  weak Byzantine agreement, defined by Lamport~\cite{Lam83}.
We give an algorithm that (a) requires only $n \geq 2f_P+1$ processes (where $f_P$ is the maximum number of faulty processes) and (b) decides
  in two delays in the common case.
With crash failures, we give the first algorithm for consensus that requires only $n \geq f_P+1$ processes and decides in two delays in the common case.
With both Byzantine or crash failures, our algorithms can also tolerate crashes of memory---only  $m \geq 2f_M+1$ memories are required, where $f_M$ is the maximum number of faulty memories.
Furthermore, with crash failures, we improve resilience further,
  to tolerate crashes of a minority of the combined set of memories and processes.

Our algorithms appear to violate known impossibility results:
it is known that with message-passing, Byzantine agreement requires
   $n \geq 3f_P+1$ even if the system is synchronous~\cite{pease1980reaching},
   while consensus with crash failures require $n \geq 2f_P+1$ if the system is
   partially synchronous~\cite{dwork1988consensus}.
There is no contradiction: our algorithms rely on the power
  of RDMA, not available in other systems.

RDMA's power comes from two features: (1) simultaneous access to message-passing and shared-memory, and
  (2) dynamic permissions.
Intuitively, shared-memory helps resilience, message-passing helps performance, and
  dynamic permissions help both. 
  
To see how shared-memory helps resilience, consider the Disk Paxos algorithm~\cite{gafni2003disk},
  which uses shared-memory (disks) but no messages.
Disk Paxos requires only $n \geq f_P+1$ processes, matching the resilience of our algorithm.
However, Disk Paxos is not as fast: it takes at least four delays.
In fact, we show that no shared-memory consensus algorithm can decide in two delays (Section~\ref{sec:imp}).
  
To see how message-passing helps performance, consider
 the Fast Paxos algorithm~\cite{lamport2006fast}, which uses message-passing and no shared-memory.
Fast Paxos decides in only two delays in common executions,
  but it requires $n \geq 2f_P+1$ processes.

Of course, the challenge is achieving both high resilience and good performance in a single
  algorithm.
This is where RDMA's dynamic permissions shine.
Clearly, dynamic permissions improve resilience against Byzantine failures, by preventing a Byzantine 
  process from overwriting memory and making it useless.
More surprising, perhaps, is that dynamic permissions help performance, by
  providing an uncontended instantaneous guarantee: if each process revokes the write permission
  of other processes before writing to a register, then a process that writes successfully knows that
  it executed uncontended, without having to take additional steps (e.g., to read the register).
We use this technique in our algorithms for both Byzantine and crash failures.

In summary, our contributions are as follows:
\begin{itemize}[noitemsep]
	\item We consider distributed systems with RDMA, and we propose a model that captures some of its key properties while accounting for failures of processes and memories, with support of dynamic permissions.

    \item We show that the shared-memory part of our RDMA improves resilience: our Byzantine agreement 
      algorithm requires only $n \geq 2f_P+1$ processes.
      
	\item We show that the shared-memory by itself does not permit consensus algorithms that
	  decide in two steps in common executions.
	  
	\item We show that with dynamic permissions, we can improve the performance of our Byzantine Agreement
	   algorithm, to decide in two steps in common executions.
	   
	\item We give similar results for the case of crash failures: decision in two steps while
	   requiring only $n \geq f_P+1$ processes.
	   
	\item Our algorithms can tolerate the failure of memories, up to a minority of them.
%
\end{itemize}

The rest of the paper is organized as follows. Section~\ref{sec:related} gives an overview of related work. In Section~\ref{sec:model} we formally define the RDMA-compliant M\&M model that we use in the rest of the paper, and specify the agreement problems that we solve.
We then proceed to present the main contributions of the paper. Section~\ref{sec:byz} presents our fast and resilient Byzantine agreement algorithm. 
In Section~\ref{sec:crash} we consider the special case of crash-only failures, and show an improvement of the algorithm and tolerance bounds for this setting. In Section~\ref{sec:imp} we briefly outline a lower bound that shows that the dynamic permissions of RDMA are necessary for achieving our results.
Finally, in Section~\ref{sec:implNotes} we discuss the semantics of RDMA in practice, and how our model reflects these features. 

\ifcamera
Due to space limitations, most proofs have been omitted from this version, and appear in the full version of this paper~\cite{fullversion}.
\else
To ease readability, most proofs have been deferred to the Appendices.
\fi

	\section{Related Work}\label{sec:related}

\xparagraph{RDMA.}
Many high-performance systems were recently proposed using RDMA, such as distributed key-value stores~\cite{dragojevic2014farm, kalia2015using}, communication primitives~\cite{dragojevic2014farm,kalia2016fasst}, and shared address spaces across clusters~\cite{dragojevic2014farm}. Kalia \textit{et al.}~\cite{kaminsky2016design}
provide guidelines for designing systems using RDMA.
RDMA has also been applied to solve consensus~\cite{behrens2016derecho, poke2015dare, wang2017apus}. Our model shares similarities with DARE~\cite{poke2015dare} and APUS~\cite{wang2017apus}, which modify queue-pair state at run time to prevent or allow 
access to memory regions, similar to our dynamic permissions. These systems perform better than TCP/IP-based solutions, by exploiting better raw performance of RDMA, without changing the fundamental communication complexity or failure-resilience of the consensus protocol.
%
Similarly, R{\"u}sch \textit{et al.}~\cite{rusch2018towards} use RDMA as a replacement for TCP/IP in existing BFT protocols.

\xparagraph{M\&M.}
Message-and-memory (M\&M) refers to a broad class of models that combine message-passing 
  with shared-memory, introduced by Aguilera et al. in~\cite{aguilera2018passing}.
In that work, Aguilera et al. consider M\&M models without memory permissions and failures, and show that such models lead to algorithms that are more robust to failures and asynchrony. In particular, they give a consensus algorithm that tolerates more crash failures than message-passing systems, but is more scalable than shared-memory systems, as well as a leader election algorithm that reduces the synchrony requirements.
In this paper, our goal is to understand how memory permissions and failures in RDMA impact agreement.



\xparagraph{Byzantine Fault Tolerance.}
Lamport, Shostak and Pease~\cite{lamport1982Byzantine,pease1980reaching} show that Byzantine agreement can be solved in synchronous systems iff $n \geq 3f_P+1$. With unforgeable signatures, Byzantine agreement can be solved iff $n \geq 2f_P+1$. 
%
%
In asynchronous systems subject to failures, consensus cannot be solved~\cite{fischer1985impossibility}. However, this result is circumvented by making
additional assumptions for liveness, such as randomization~\cite{ben1983another} or partial synchrony~\cite{chandra1996unreliable,dwork1988consensus}.
Many Byzantine agreement algorithms focus on safety and implicitly use the
  additional assumptions for liveness.
Even with signatures, asynchronous Byzantine agreement can be solved only if $n\geq 3f_P+1$~\cite{bracha1985asynchronous}.

It is well known that the resilience of Byzantine agreement varies depending on various model assumptions like synchrony, signatures, equivocation, and the exact variant of the problem to be solved. A system that has non-equivocation is one that can prevent a Byzantine process from sending different values to different processes. Table~\ref{tab:related} summarizes some known results that are relevant to this paper. 


{\footnotesize
\begin{table}[h!]
	\begin{tabular}{cccccc}
	\toprule
		\textbf{Work} & \textbf{Synchrony} & \textbf{Signatures} & \textbf{Non-Equiv} & \begin{tabular}[x]{@{}c@{}}\textbf{Strong}\\\textbf{Validity}\end{tabular} & \textbf{Resiliency} \\
		\midrule
\cite{lamport1982Byzantine}			& \ding{51}             & \ding{51}             & \ding{55}                    & \ding{51}                   & $2f+1$                   \\
\cite{lamport1982Byzantine}		& \ding{51}             & \ding{55}              & \ding{55}                    & \ding{51}                   & $3f+1$                   \\
\cite{malkhi2003objects,alon2005tight}		& \ding{55}             & \ding{51}              & \ding{51}                    & \ding{51}                   & $3f+1$                   \\
	\cite{clement2012limited}		& \ding{55}             & \ding{51}              & \ding{55}                    & \ding{55}                   & $3f+1$                   \\
	\cite{clement2012limited}		& \ding{55}             & \ding{55}              & \ding{51}                    & \ding{55}                   & $3f+1$                   \\
		\cite{clement2012limited}		& \ding{55}             & \ding{51}              & \ding{51}                    & \ding{55}                   & $2f+1$    \\
		\midrule
	    This paper & \ding{55}             & \ding{51}              & \begin{tabular}[x]{@{}c@{}}\ding{55}\\(RDMA)\end{tabular}                     & \ding{55}                   & $2f+1$ \\
	    \bottomrule
	\end{tabular}
	    \caption{Known fault tolerance results for Byzantine agreement.}
	    	    \label{tab:related}
\end{table}
}

Our Byzantine agreement results share similarities with results for shared memory.
Malkhi \textit{et al.}~\cite{malkhi2003objects} and Alon \textit{et al.}~\cite{alon2005tight} show bounds on the resilience of strong and weak consensus in a model with reliable memory but Byzantine processes. They also provide consensus protocols, using read-write registers enhanced with sticky bits (write-once memory) and access control lists not unlike our permissions. Bessani \textit{et al.}~\cite{bessani2009sharing} propose an alternative to sticky bits and access control lists through Policy-Enforced Augmented Tuple Spaces.
All these works handle Byzantine failures with powerful objects rather than registers.
Bouzid \textit{et al.}~\cite{bouzid2016necessary} show that $3f_P+1$ processes are necessary for strong Byzantine agreement with read-write registers.

Some prior work solves Byzantine agreement with $2f_P{+}1$ processes
  using specialized trusted components that
  Byzantine processes cannot control~\cite{ChunMSK07,CorreiaNV04,KapitzaBCDKMSS12,VeroneseCBLV13,ChunMS08,correia2010asynchronous}.
Some schemes decide in two delays but require a large
  trusted component: a coordinator~\cite{ChunMS08}, reliable
  broadcast~\cite{correia2010asynchronous}, or message ordering~\cite{KapitzaBCDKMSS12}.
For us, permission checking in RDMA is a trusted component of sorts, but it is small  and readily available.

At a high-level, our improved Byzantine fault tolerance is achieved by preventing equivocation by Byzantine processes, thereby effectively translating each Byzantine failure into a crash failure. Such translations from one type of failure into a less serious one have appeared extensively in the literature~\cite{clement2012limited,bracha1985asynchronous,neiger1990automatically,bazzi1991optimally}. Early work~\cite{neiger1990automatically,bazzi1991optimally} shows how to translate a crash tolerant algorithm into a Byzantine tolerant algorithm in the synchronous setting.
Bracha~\cite{bracha1987asynchronous} presents a similar translation for the asynchronous setting, in which $n \geq 3f_P+1$ processes are required to tolerate $f_P$ Byzantine failures. Bracha's translation relies on the definition and implementation of a reliable broadcast primitive, very similar to the one in this paper. However, we show that using the capabilities of RDMA, we can implement it with higher fault tolerance.

\xparagraph{Faulty memory.}
Afek \textit{et al.}~\cite{AGMT1992} and Jayanti \textit{et al.}~\cite{JCT1998} study the problem of masking the benign failures of shared memory or objects.
We use their ideas of replicating data across memories.
Abraham \textit{et al.}~\cite{abraham2006Byzantine} considers honest processes but malicious memory. 

\xparagraph{Common-case executions.}
Many systems and algorithms tolerate adversarial scheduling but optimize for
  common-case executions without failures, asynchrony, contention, etc (e.g., \cite{dobre2006one,dutta2005fast,keidar2001cost,kursawe2002,lamport2006fast,MA2006,fastpaxos}).
None of these match both the resilience and performance of our algorithms.
Some algorithms decide in one delay but require $n \ge 5f_P+1$ for Byzantine failures~\cite{bosco} or $n\ge 3f_P+1$ for crash failures~\cite{dobre2006one,brasileiro2001consensus}.

	\section{Model and Preliminaries} \label{sec:model}

We consider a message-and-memory (M\&M) model, which allows processes to use
  both message-passing and shared-memory~\cite{aguilera2018passing}.
The system has $n$ processes $P= \{p_1,\ldots, p_{n}\}$
  and $m$ (shared) \emph{\disks} $M = \{ \mu_1, \ldots, \mu_m \}$.
Processes communicate by accessing \disks{} or sending messages.
Throughout the paper, \disk{} refers to the shared \disks{}, not the local state of processes.

The system is asynchronous in that it can experience arbitrary delays.
We expect algorithms to satisfy the safety properties of the problems we consider, under
  this asynchronous system.
For liveness, we require additional standard assumptions, such as
  partial synchrony, randomization, or failure detection.


\begin{figure}
\centering
\includegraphics[width=3in,angle=270,trim=0 100 0 0, clip]{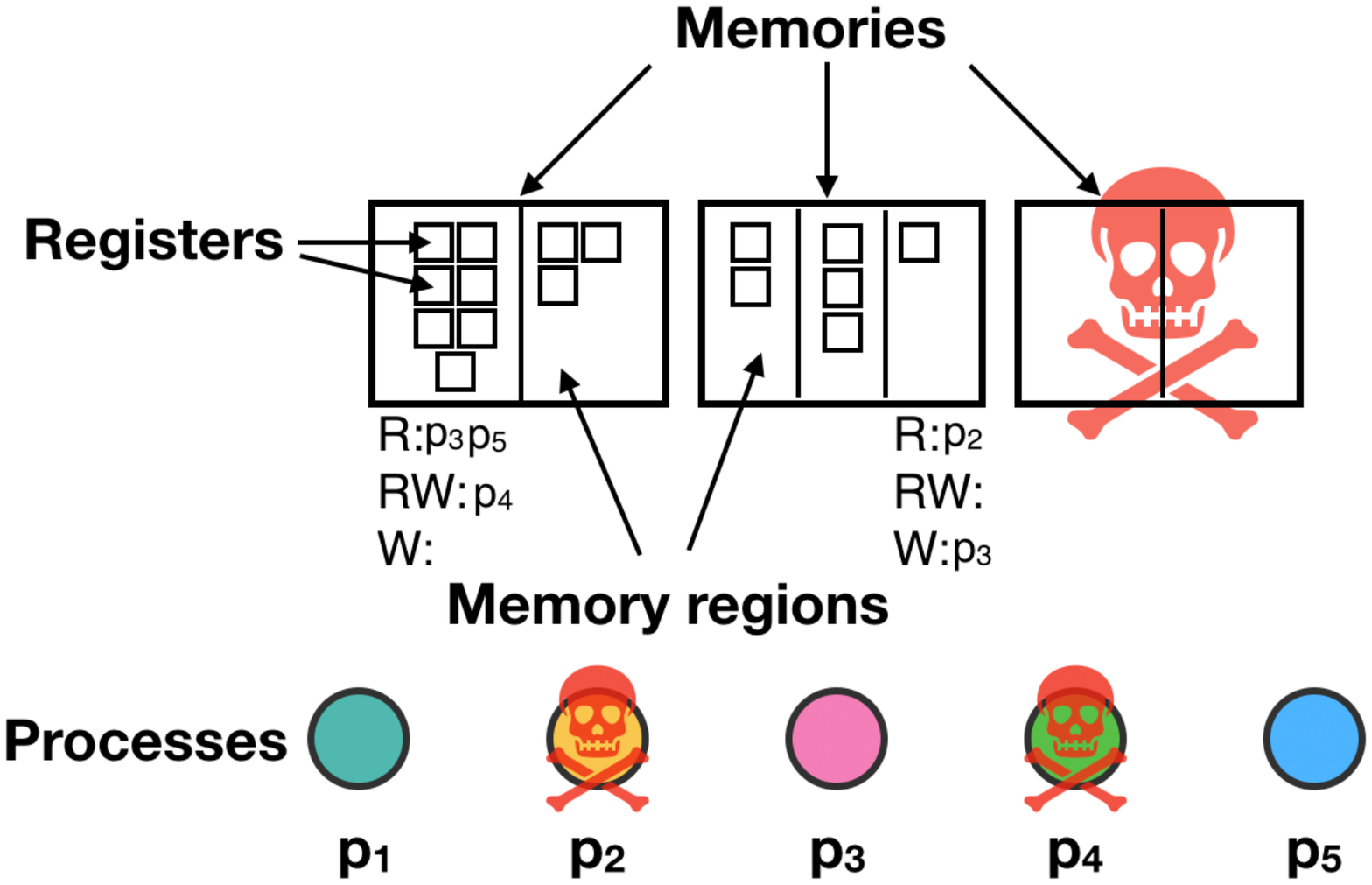}
\caption{Our model with processes and memories, which may both fail.
 Processes can send messages
 to each other or access registers in the memories. 
Registers in a memory are grouped into memory regions that may overlap,
   but in our algorithms they do not. Each region has a permission
   indicating what processes can read, write, and read-write
   the registers in the region (shown for two regions).
  }
\label{fig:myfig}
\end{figure}

\textbf{Memory permissions.}
Each \disk{} consists of a set of \emph{registers}.
To control access,
  an algorithm groups those registers into a set of (possibly overlapping)
  \emph{\mregion{s}}, and then defines permissions for those \mregion{s}.
Formally, a \mregion{} $\mr$ of a \disk{} $\mu$ is a subset of the registers of $\mu$.
We often refer to $\mr$ without specifying the \disk{} $\mu$ explicitly.
Each \mregion{} $\mr$ has a \emph{permission}, which consists of
%
%
three disjoint sets of processes $R_{\mr}$, $W_{\mr}$, $\RW_{\mr}$ indicating whether each process can read, write, or read-write the registers in the region.
We say that \emph{$p$ has read permission on $\mr$} if $p\in R_\t{mr}$ or $p\in\RW_\t{mr}$;
  we say that \emph{$p$ has write permission on $\mr$} if $p \in W_\t{mr}$ or $p \in\RW_\t{mr}$.
In the special case when $R_{\mr}=P\setminus\{p\}$, $W_{\mr}=\emptyset$, $\RW_{\mr}=\{p\}$, we say that
  $\mr$ is a Single-Writer Multi-Reader (SWMR) region---registers in $\mr$ correspond to the
  traditional notion of SWMR registers.
%
%
Note that a register may belong to several regions, and a process may have access
  to the register on one region but not another---this models the existing \rdma
  behavior.
Intuitively, when reading or writing data, a process specifies the region and the
  register, and
  the system uses the region to determine if access is allowed (we make this precise
  below).

\textbf{Permission change.}
An algorithm indicates an initial permission for each \mregion{} $\mr$.
Subsequently, the algorithm may wish to change the permission of $\mr$ during execution.
For that, processes can invoke an operation $\t{changePermission}(\mr,\t{new\_perm})$,
  where $\t{new\_perm}$ is a triple $(R,W,\RW)$.
This operation returns no results and it is intended to modify $R_{\mr},W_{\mr},\RW_{\mr}$ to
  $R,W,\RW$.
To tolerate Byzantine processes, an algorithm can restrict
  processes from changing permissions.
For that, the algorithm specifies
  a function $\t{legalChange}(p, \mr, \t{old\_perm}, \t{new\_perm})$
  which returns a boolean indicating whether process $p$ can change the
  permission of $\mr$ to $\t{new\_perm}$ when the current permissions are
  $\t{old\_perm}$.
More precisely, when $\t{changePermission}$ is invoked, the system
  evaluates $\t{legalChange}$ to determine whether $\t{changePermission}$ 
  takes effect or becomes a no-op.
When $\t{legalChange}$ always returns false, we say that the \emph{permissions are
  static}; otherwise, the \emph{permissions are dynamic}.

%
%

\textbf{Accessing \disks.} Processes access the \disks{} via operations
  $\t{write}(\mr,r,v)$ and $\t{read}(\mr, r)$ for \mregion{} $\mr$, register $r$, and
  value $v$.
A $\t{write}(\mr,r,v)$ by process $p$ changes register $r$ to $v$ and returns $\ack$
  if $r \in \mr$ and $p$ has write permission on $\mr$; otherwise, the operation
  returns $\nak$.
A $\t{read}(\mr,r)$ by process $p$ returns the last value successfully written
  to $r$ if $r \in \mr$ and $p$ has read permission on $\mr$; otherwise, the
  operation returns $\nak$.
In our algorithms, a register belongs to exactly one region, so we omit
  the $\mr$ parameter from write and read operations.


\textbf{Sending messages.} Processes can also communicate by sending messages over a set of directed links.
We assume messages are unique.
If there is a link from process $p$ to process $q$, then $p$
  can send messages to $q$.
Links satisfy two properties: \emph{integrity} and \emph{no-loss}.
Given two correct processes $p$ and $q$,
  integrity requires that a message $m$ be received by $q$ from $p$ at most once
  and only if $m$ was previously sent by $p$ to $q$.
No-loss requires that a message $m$ sent from $p$ to $q$ be
  eventually received by $q$.
In our algorithms, we typically assume a fully connected
  network so that every pair of correct processes can communicate.
We also consider the special case when there are no links
  (see below).

\textbf{Executions and steps.} An execution is as a sequence of process steps.
In each step, a process does the following, according to its local state:
   (1) sends a message or invokes an operation on a \disk{}
       (read, write, or changePermission), 
   (2) tries to receive a message or a response from an outstanding operation, and
   (3) changes local state.
We require a process to have at most one outstanding operation on
  each \disk.


\textbf{Failures.} A memory $m$ may fail by crashing,
  which causes subsequent operations on its registers to
  hang without returning a response.
Because the system is asynchronous, a process cannot differentiate
  a crashed memory from a slow one.
We assume there is an upper bound $f_M$ on the maximum number of
  \disks{} that may crash.
%
Processes may fail by crashing or
  becoming Byzantine.
If a process crashes, it stops taking steps forever.
If a process becomes Byzantine, it can deviate
  arbitrarily from the algorithm.
However, that process cannot operate on \disks{} without
  the required permission.
We assume there is an upper bound $f_P$ on the maximum number of
  processes that may be faulty.
  Where the context is clear, we omit the $P$ and $M$ subscripts from the number of failures, $f$.

\textbf{Signatures.}
Our algorithms assume unforgeable signatures:
  there are primitives $\t{sign}(v)$ and $\t{sValid}(p,v)$ which, respectively,
  signs a value $v$ and determines if $v$ is signed by process $p$.

\textbf{Messages and disks.}
The model defined above includes two common models
  as special cases.
In the \emph{message-passing} model, there are no \disks{} ($m=0$),
  so processes can communicate only by sending messages.
In the \emph{disk} model~\cite{gafni2003disk}, there are no
  links, so processes can communicate only via \disks;
  moreover, each \disk{} has a single region which always
  permits all processes to read and write all registers.

\subsection*{Consensus}

In the consensus problem, processes propose an initial value and
  must make an irrevocable decision on a value.
With crash failures, we require the following properties:
\begin{itemize}[noitemsep, nolistsep]
    \item \textbf{Uniform Agreement.} If processes $p$ and $q$ decide $v_p$ and $v_q$, then $v_p = v_q$.
    \item \textbf{Validity.} If some process decides $v$, then $v$ is the initial value proposed by some process.
    \item \textbf{Termination.} Eventually all correct processes decide.
\end{itemize}

We expect Agreement and Validity to hold in an asynchronous system, while Termination
  requires standard additional assumptions (partial synchrony, randomization, failure
  detection, etc).
With Byzantine failures, we change these definitions so the problem can be solved.
We consider weak Byzantine agreement~\cite{Lam83}, with the following properties:

\rmv{
    In a Byzantine setting, the requirements must be adjusted a little.  In particular, we cannot require that all agents that decide a value agree, since Byzantine agents may decide arbitrary values. Similarly, we cannot require that Byzantine agents output valid decision values. However, we do allow a correct agent to decide on a value that was proposed by a Byzantine proposer. This requirement is known as \emph{weak validity}.
}    

\begin{itemize}[noitemsep, nolistsep]
	\item \textbf{Agreement.} If correct processes $p$ and $q$ decide $v_p$ and $v_q$, then $v_p = v_q$.
	\item \textbf{Validity.}  With no faulty processes, if some process decides $v$, then $v$ is the input of some process.
    \item \textbf{Termination.} Eventually all correct processes decide.
\end{itemize}


\textbf{Complexity of algorithms.}
We are interested in the performance of algorithms in \emph{common-case executions}, 
  when the system is synchronous and there are no failures.
In those cases, we measure performance using the notion of \emph{delays}, which
  extends message-delays to our model.
Under this metric, computations are instantaneous, 
  each message takes one delay, and each memory operation takes two delays.
Intuitively, a delay represents the time incurred by the network to transmit a message;
  a memory operation takes two delays because its hardware implementation requires a round trip.
%
We say that a consensus protocol is \emph{$k$-deciding} if, in common-case executions,
  some process decides in $k$ delays.

	\section{Byzantine Failures}\label{sec:byz}

We now consider Byzantine failures and give a 2-deciding algorithm for weak Byzantine agreement
  with $n \ge 2f_P+1$ processes and $m \ge 2f_M + 1$ \disks{}. 
The algorithm consists of the composition of two sub-algorithms: a slow one that
  always works, and a fast one that gives up under hard conditions.

The first sub-algorithm, called \emph{\robustbyz}, is developed in two steps.
We first implement a \emph{reliable broadcast} primitive, 
  which prevents Byzantine processes from sending different values to different processes.
Then, we use the framework of Clement \textit{et al.}~\cite{clement2012limited} combined
  with this primitive to convert a message-passing
  consensus algorithm that tolerates crash failures
  into a consensus algorithm that tolerates Byzantine failures.
This yields \robustbyz.\footnote{The attentive reader may wonder why at this point we have not
        achieved a 2-deciding algorithm already: if we apply
        Clement \textit{et al.}~\cite{clement2012limited} to a 2-deciding crash-tolerant
        algorithm (such as Fast Paxos~\cite{fastpaxos}),
        will the result not be a 2-deciding Byzantine-tolerant algorithm?
        The answer is no, because Clement \textit{et al.} needs reliable broadcast,
        which incurs at least 6 delays.}
It uses only static permissions and assumes \disks{} are split into SWMR regions.
Therefore, this sub-algorithm works in the 
  traditional shared-memory model with SWMR registers, and it may be of independent interest.

\rmv{
    In this section we present a Byzantine fault tolerant consensus protocol for our model. Our protocol solves weak Byzantine agreement with $2f_P+1$ processes and $2f_M + 1$ \disks{}, and is 2-deciding. We present the protocol in two parts, which can be combined modularly to form the final algorithm. 
    Before presenting either part, we first show how to implement a key primitive that helps us achieve the high robustness of our algorithm. This primitive is called \emph{non-equivocating broadcast}. Intuitively, this primitive behaves like a classic broadcast primitive, but prevents Byzantine processes from sending different values to different processes. With non-equivocation broadcast, we can use the framework of Clement \textit{et al.}~\cite{clement2012limited} to essentially convert Byzantine failures into crash failures.
    This protocol uses static permissions only; it assumes all \disks{} are split into singe-writer multi-reader (SWMR) regions. Since it does not use dynamic permissions, this protocol can be implemented in classic SWMR shared memory (where memory does not fail), and may be of independent interest. 
}

The second sub-algorithm is called \emph{\fastbyz}.
It uses dynamic permissions to decide
  in two delays using one signature in common executions.
However, the sub-algorithm gives up if the system is not synchronous or there are Byzantine
  failures.

Finally, we combine both sub-algorithms using ideas from the
  Abstract framework of Aublin et al.~\cite{aublin2015next}.
More precisely, we start by running \fastbyz{}; if it aborts, we run \robustbyz{}.
There is a subtlety: for this idea to work, \robustbyz{} must decide on a value $v$ if 
  \fastbyz{} decided $v$ previously.
To do that, \robustbyz{} decides on a \emph{preferred value} if at least $f+1$ processes 
  have this value as input. 
To do so, we use the classic crash-tolerant Paxos algorithm (run under the \robustbyz{} algorithm to ensure Byzantine tolerance) but with an initial set-up phase 
  that ensures this safe decision.
We call the protocol \emph{\prefPax}.



\subsection{The \robustbyz{} Sub-Algorithm}\label{sec:robustProtocol}

%
We develop \robustbyz{} using the construction by 
  Clement \textit{et al.}~\cite{clement2012limited}, which we now explain.
Clement \textit{et al.} show how to transform a message-passing
  algorithm $\cal A$ that tolerates $f_P$ crash failures
  into a
  message-passing algorithm that tolerates $f_P$ Byzantine failures
  in a system where $n \ge 2f_P+1$ processes, assuming
  unforgeable signatures and a non-equivocation mechanism.
They do so by implementing trusted message-passing primitives, \emph{T-send} and \emph{T-receive}, using non-equivocation and signature verification on every
  message.
Processes include their full history with each message, and then verify locally whether a received message is consistent with the protocol.  This restricts Byzantine behavior to crash
failures.

To apply this construction in our model, we show that
  our model can implement non-equivocation and message passing.
We first show that shared-memory with SWMR registers (and no memory failures)
  can implement these primitives, and then show
  how our model can implement shared-memory with SWMR registers.
  
\subsubsection{Reliable Broadcast}
Consider a shared-memory system. We present a way to prevent equivocation through a solution to the \textit{reliable broadcast} problem, which we recall below. Note that our definition of reliable broadcast includes the sequence number $k$ (as opposed to being single-shot) so as to facilitate the integration with the Clement et al. construction, as we explain in Section~\ref{sec:applying-clement}.

\def\broadcast{\t{broadcast}}
\def\deliver{\t{deliver}}

\begin{definition}
	\emph{Reliable broadcast} is defined in terms of two primitives, 
	$\broadcast(k,m)$ and $\deliver(k,m,q)$.
	When a process $p$ invokes $\broadcast(k,m)$ we say that \emph{$p$ broadcasts
	$(k,m)$}. When a process $p$ invokes $\deliver(k,m,q)$ we say that
	\emph{$p$ delivers $(k,m)$ from $q$}.
	Each correct process $p$ must invoke $\broadcast(k,*)$ with $k$ one higher than $p$'s previous
	invocation (and first invocation with $k{=}1$).
	The following holds:
	\begin{enumerate}[noitemsep,nolistsep]
		\item If a correct process $p$ broadcasts $(k,m)$, then all correct 
		processes eventually deliver $(k,m)$ from $p$.
		\item 	If $p$ and $q$ are correct processes, $p$ delivers $(k,m)$ from $r$, and $q$ delivers $(k,m')$ from $r$, then $m{=}m'$.
		\item If a correct process delivers $(k,m)$ from a correct process $p$, then $p$ must have broadcast $(k,m)$.
		\item If a correct process delivers $(k,m)$ from $p$, then all correct processes eventually deliver $(k,m')$ from $p$ for some $m'$.
	\end{enumerate}
\end{definition}


\renewcommand{\figurename}{Algorithm}
\begin{figure}
	\caption{Reliable Broadcast}
	\begin{lstlisting}[columns=fullflexible,breaklines=true, keywords={}]
SWMR Value[n,M,n]; initialized to @$\bot$@. Value[p] is array of SWMR(p) registers.

SWMR L1Proof[n,M,n]; initialized to @$\bot$@. L1Proof[p] is array of SWMR(p) registers.

SWMR L2Proof[n,M,n]; initialized to @$\bot$@. L2Proof[p] is array of SWMR(p) registers.

Code for process p
	last[n]: local array with last k delivered from each process. Initially, last[q] = 0
	state[n]: local array of registers. state[q] @$\in$@ {WaitForSender,WaitForL1Proof,WaitForL2Proof}. Initially, state[q] =  WaitForSender
	
    broadcast (k,m){
        Value[p,k,p].write(sign((k,m))); @\label{line:broadcastWrite}@   }
	   
    value checkL2proof(q,k) {
        for i@$\in \Pi$@ { 
            proof = L2Proof[i,k,q].read();@\label{line:readL2proof}@
            if (proof != @$\bot$@ && check(proof)) {
                L2Proof[p,k,q].write(proof);@\label{line:copyL2proof}@
                return proof.msg; }   }
        return null;    }
	
	for q in @$\Pi$@ in parallel {@\label{line:forloop}@
        while true { 
            try_deliver(q);	@\label{line:tryDeliver}@}} 	}
		
	try_deliver(q) {
        k = last[q];
        val = checkL2Proof(q,k);
        if (val != null) {
            deliver(k, proof.msg, q);@\label{line:deliver}@
            last[q] += 1;
            state = WaitForSender;
            return;
        }
        
        if state == WaitForSender {
            val = Value[q,k,q].read();@\label{line:readMessage}@
            if (val==@$\bot$@ || !sValid(p, val) || key!=k)
                return;
            Value[p,k,q].write(sign(val));@\label{line:copyVal}@
            state = WaitForL1Proof; }
        
        if state == WaitForL1Proof {
            checkedVals = @$\emptyset$@;
            for i @$\in \Pi$@{
                val = Value[i,k,q].read();@\label{line:readValCopies}@
                if (val!=@$\bot$@ && sValid(p,val) && key==k) {
                    add val to checkedVals; }   }
            
            if size(checkedVals) @$\geq$@ majority and checkedVals contains only one value {@\label{line:checkL1proof}@
                l1prf = sign(checkedVals);
                L1Proof[p,k,q].write(l1prf);@\label{line:writel1prf}@
                state = WaitForL2Proof; }   }
        
        if state == WaitForL2Proof{
            checkedL1Proofs = @$\emptyset$@;
            for i in @$\Pi$@{
                proof = L1Proof[i,k,q].read();@\label{line:readL1prf}@
                if ( checkL1Proof(proof) ) {
                    add proof to checkedL1Proofs; }   }
            
            if size(checkedL1Proofs) @$\geq$@ majority {
                l2prf = sign(checkedL1Proofs);
                L2Proof[p,k,q].write(l2prf); }   } }
	\end{lstlisting}
	\label{alg:non-eq}
\end{figure}
\renewcommand{\figurename}{Figure}


Algorithm~\ref{alg:non-eq} shows how to implement reliable broadcast that is tolerant to a minority of Byzantine failures in shared-memory using SWMR registers.

To broadcast its $k$-th message $m$, $p$ simply signs $(k,m)$ and writes it in slot $Value[p,k,p]$ of its memory\footnote{The indexing of the slots is as follows: the first index is the writer of the SWMR register, the second index is the sequence number of the message, and the third index is the sender of the message.}.

Delivering a message from another process is more involved, requiring verification steps to ensure that 
all correct processes will eventually deliver the same message and no other.
The high-level idea is that before delivering a message $(k,m)$ from $q$, each process $p$ checks that no other process saw a different value from $q$, and waits to hear that ``enough'' other processes also saw the same value.
More specifically, each process $p$ 
has 3 slots per process per sequence number, that only $p$ can write to, but all processes can read from. These slots are initialized to $\bot$, and $p$ uses them to write the values that it has seen. The 3 slots represent 3 levels of `proofs' that this value is correct; for each process $q$ and sequence number $k$, $p$ has a slot to write (1) the initial value $v$ it read from $q$ for $k$, (2) a proof that at least $f+1$ processes saw the same value $v$ from $q$ for $k$, and (3) a proof that at least $f+1$ processes wrote a proof of seeing value $v$ from $q$ for $k$ in their second slot. We call these slots the Value slot, the L1Proof slot, and the L2Proof slot, respectively.

We note that each such valid proof has signed copies of only one value for the message. Any proof that shows copies of two different values or a value that is not signed is not considered valid. If a proof has copies of only value $v$, we say that this proof \emph{supports $v$}.

To deliver a value $v$ from process $q$ with sequence number $k$, process $p$ must successfully write a valid proof-of-proofs in its L2Proof slot supporting value $v$ (we call this an L2 proof). It has two options of how to do this; firstly, if it sees a valid L2 proof in some other process $i$'s $L2Proof[i,k,q]$ slot, it copies this proof over to its own L2 proof slot, and can then deliver the value that this proof supports. If $p$ does not find a valid L2 proof in some other process's slot, it must try to construct one itself. We now describe how this is done.

A correct process $p$ goes through three stages when constructing a valid L2 proof for $(k,m)$ from $q$. In the pseudocode, the three stages are denoted using states that $p$ goes through: WaitForSender, WaitForL1Proof, and WaitForL2Proof.

In the first stage, WaitForSender, $p$ reads $q$'s $Value[q,k,q]$ slot. If $p$ finds a $(k,m)$ pair, $p$ signs and copies it to its $Value[p,k,q]$ slot and enters the WaitForL1Proof state.

In the second stage, WaitForL1Proof, $p$ reads all $Value[i,k,q]$ slots, for $i\in\Pi$. If all the values $p$ reads are correctly signed and equal to $(k,m)$, and if there are at least $f+1$ such values, then $p$ compiles them into an L1 proof, which it signs and writes to $L1Proof[p,k,q]$; $p$ then enters the WaitForL2Proof state. 

In the third stage, WaitForL2Proof, $p$ reads all $L1Proof[i,k,q]$ slots, for $i\in\Pi$. If $p$ finds at least $f+1$ valid and signed L1 proofs for $(k,m)$, then $p$ compiles them into an L2 proof, which it signs and writes to $L2Proof[p,k,q]$. The next time that $p$ scans the $L2Proof[\cdot,k,q]$ slots, $p$ will see its own L2 proof (or some other valid proof for $(k,m)$) and deliver $(k,m)$.

This three-stage validation process ensures the following crucial property: no two valid L2 proofs can support different values. Intuitively, this property is achieved because for both L1 and L2 proofs, at least $f+1$ values of the previous stage must be copied, meaning that at least one correct process was involved in the quorum needed to construct each proof. Because correct processes read the slots of \emph{all} others at each stage before constructing the next proof, and because they never overwrite or delete values that they already wrote, it is guaranteed that no two correct processes will create valid L1 proofs for different values, since one must see the Value slot of the other. Thus, no two processes, Byzantine or otherwise, can construct valid L2 proofs for different values.

Notably, a weaker version of broadcast, which does not require Property 4 (i.e., Byzantine consistent broadcast~\cite[Module 3.11]{RachidBook}), can be solved with just the first stage of Algorithm~\ref{alg:non-eq}, without the L1 and L2 proofs. The purpose of those proofs is to ensure the 4th property holds; that is, to enable all correct processes to deliver a value once some correct process delivered.


In the appendix, we formally prove the above intuition and arrive at the following lemma.

\begin{lemma}\label{lem:SMnon-eq}
	Reliable broadcast can be solved in shared-memory with SWMR regular registers with $n\geq 2f+1$ processes.
\end{lemma}

\subsubsection{Applying Clement et al's Construction}
\label{sec:applying-clement}
Clement et al. show that given unforgeable transferable signatures and non-equivocation, one can reduce Byzantine failures to crash failures in message passing systems~\cite{clement2012limited}. They define non-equivocation as a predicate $valid_p$ for each process $p$, which takes a sequence number and a value and evaluates to true for just one value per sequence number. All processes must be able to call the same $valid_p$ predicate, which always terminates every time it is called.

We now show how to use reliable broadcast to implement messages with transferable signatures and non-equivocation as defined by Clement et al.~\cite{clement2012limited}. Note that our reliable broadcast mechanism already involves the use of transferable signatures, so to send and receive signed messages, one can simply use broadcast and deliver those messages. However, simply using broadcast and deliver is not enough to satisfy the requirements of the $valid_p$ predicate of Clement et al. The problem occurs when trying to validate nested messages recursively.

In particular, recall that in Clement et al's construction, whenever a message is sent, the entire history of that process, including all messages it has sent and received, is attached. Consider two Byzantine processes $q_1$ and $q_2$, and assume that $q_1$ attempts to equivocate in its $k$th message, signing both $(k,m)$ and $(k,m')$. Assume therefore that no correct process delivers any message from $q_1$ in its $k$th round. However, since $q_2$ is also Byzantine, it could claim to have delivered $(k,m)$ from $q_1$. If $q_2$ then sends a message that includes $(q,k,m)$ as part of its history, a correct process $p$ receiving $q_2$'s message must recursively verify the history $q_2$ sent. To do so, $p$ can call \texttt{try\_deliver} on $(q_1,k)$. However, since no correct process delivered any message from $(q_1,k)$, it is possible that this call never returns.

To solve this issue, we introduce a \texttt{validate} operation that can be used along with broadcast and deliver to validate  the correctness of a given message. The \texttt{validate} operation is very simple: it takes in a process id, a sequence number, and a message value $m$, and simply runs the \texttt{checkL2proof} helper function. If the function returns a proof supporting $m$, \texttt{validate} returns true. Otherwise it returns false. The pseudocode is shown in Algorithm~\ref{alg:validate}.

\renewcommand{\figurename}{Algorithm}
\begin{figure}
	\caption{Validate Operation for Reliable Broadcast}
\begin{lstlisting}[columns=fullflexible,breaklines=true, keywords={}]
bool validate(q,k,m){
    val = checkL2proof(q,k);
    if (val == m) {
        return true;    }
    return false;   }
\end{lstlisting}
\label{alg:validate}
\end{figure}

In this way, Algorithms~\ref{alg:non-eq} and~\ref{alg:validate} together provide signed messages and a non-equivocation primitive. Thus, combined with the construction of Clement  et al.~\cite{clement2012limited}, we immediately get the following result.

\begin{theorem} 
There exists an algorithm for weak Byzantine agreement in a shared-memory system with SWMR regular registers, signatures, and
  up to $f_P$ process crashes where $n \ge 2f_P+1$.
\end{theorem}

\paragraph{Non-equivocation in our model.}
To convert the above algorithm to our model, where memory may fail,
we use the ideas in~\cite{attiya1995sharing,AGMT1992,JCT1998} to
implement failure-free SWMR regular registers from the fail-prone
memory, and then run weak Byzantine agreement using those
regular registers. To implement an SWMR register,
a process writes or reads all \disks, and waits for a majority
  to respond.
When reading, if $p$ sees exactly one distinct non-$\bot$ value $v$ across
  the \disks, it returns $v$; otherwise, it returns $\bot$.

\begin{definition}
    Let $\mathcal{A}$ be a message-passing algorithm. \emph{\robustbyz($\mathcal{A}$)} is the algorithm $\mathcal{A}$ in which all $send$ and $receive$ operations are replaced by T-send and T-receive operations (respectively) implemented with reliable broadcast.
\end{definition}

\ifcamera
From the result of Clement et al.~\cite{clement2012limited}, Lemma~\ref{lem:SMnon-eq}, and the above handling of memory failures, it is easy to see that with $\mathcal{A}$ being a correct consensus algorithm for the crash-only setting, \robustbyz($\mathcal{A}$) solves weak Byzantine agreement with the desired fault tolerance in our dynamic permission M\&M model. This is summarized in the follow theorem.
\else
Thus we get the following lemma, from the result of Clement et al.~\cite{clement2012limited}, Lemma~\ref{lem:SMnon-eq}, and the above handling of memory failures.

\begin{lemma}\label{lem:robustByz}
    If $\mathcal{A}$ is a consensus algorithm that is tolerant to $f$ process crash failures, then \robustbyz($\mathcal{A}$) is a weak Byzantine agreement algorithm that is tolerant to up to $f_P$ Byzantine processes and $f_M$ memory crashes,
    where $n \ge 2f_P+1$ and $m \ge 2f_M+1$ in the message-and-memory model.
\end{lemma}

The following theorem is an immediate corrolary of the lemma.
\fi
\begin{theorem}
There exists an algorithm for Weak Byzantine Agreement in a message-and-memory model with up to $f_P$ Byzantine processes and $f_M$ memory crashes,
  where $n \ge 2f_P+1$ and $m \ge 2f_M+1$.
\end{theorem}


\subsection{The \fastbyz{} Sub-Algorithm}\label{sec:fastbyz}

\renewcommand{\figurename}{Algorithm}
\begin{figure}
	\caption{\fastbyz{} normal operation---code for process $p$}
\begin{lstlisting}[columns=fullflexible,breaklines=true, keywords={}]
Leader code
propose(v) {
    sign(v);
    status = Value[@$\ell$@].write(v); @\label{line:propose}@
    if (status == nak) Panic_mode();
    else decide(v); @\label{line:leaderDecide}@}

Follower code
propose(w){
    do {v = read(Value[@$\ell$@]);
        for all q @$\in \Pi$@ do pan[q] = read(Panic[q]);
    } until (v @$ \neq \bot$@ || pan[q] == true for some q || timeout); @\label{line:timeout1}@
    if (v @$ \neq \bot$@ && sValid(p1,v)) { @\label{line:checkLeader}@
        sign(v);
        write(Value[p],v);@\label{line:replicate}@
        do {for all q @$\in \Pi$@ do val[q] = read(Value[q]);
            if |{q : val[q] == v}| @$\ge$@ n then {
                Proof[p].write(sign(val[1..n]));
                for all q @$\in \Pi$@ do prf[q] = read(Proof[q]);
                if |{q : verifyProof(prf[q]) == true}| @$\ge$@ n { decide(v); exit; } @\label{line:countCheck}@}
            for all q @$\in \Pi$@ do pan[q] = read(Panic[q]);
        } until (pan[q] == true for some q || timeout); @\label{line:timeout2}@}
    Panic_mode();}
\end{lstlisting}
\label{alg:normal}
\end{figure}
\renewcommand{\figurename}{Figure}

%


\renewcommand{\figurename}{Algorithm}
\begin{figure}
	\caption{\fastbyz{} panic mode---code for process $p$}
\begin{lstlisting}[breaklines=true, keywords={}]
panic_mode(){
    Panic[p] = true;
    changePermission(Region[@$\ell$@], R: @$\Pi$@, W: {}, RW: {});  // remove write permission @\label{line:revokeLeader}@
    v = read(Value[p]); @\label{line:readOwn}@
    prf = read(Proof[p]);
    if (v @$\neq \bot$@){  Abort with <v, prf>; return; } @\label{line:abortWithOwn}@ @\label{line:myval}@
    LVal = read(Value[@$\ell$@]); 
    if (LVal @$\neq \bot$@) {Abort with <LVal, @$\bot$@>; return;}@\label{line:LVal}@
    Abort with <myInput, @$\bot$@>; }
\end{lstlisting}
\label{alg:panic}
\end{figure}
\renewcommand{\figurename}{Figure}

%

We now give an algorithm that decides in two delays in common executions in which the system is synchronous and there are no failures.
It requires only
 one signature for a fast decision, whereas the
best prior algorithm requires $6f_P + 2$ signatures and $n \ge 3f_P+1$~\cite{aublin2015next}.
%
Our algorithm, called \fastbyz{}, is not in itself a complete consensus algorithm; it may abort in some executions.
%
If \fastbyz{} aborts, it outputs an \textit{abort value}, which is used to initialize the  \robustbyz{} so that their composition preserves weak Byzantine agreement.
This composition is inspired by the Abstract framework of Aublin \textit{et al.}~\cite{aublin2015next}.

%

\rmv{
    The algorithm uses the framework of Aublin \textit{et al.}~\cite{aublin2015next}, which consists of
      a sequence of consensus instances, such that if an instance aborts, the next instance
      is invoked.
    \fastbyz{} is an algorithm for the first instance, which provides the required performance
      for common-case executions. 
    We show that the robust non-aborting algorithm presented in Section~\ref{sec:robustProtocol} can serve as another instance of Abstract. Thus the two together yield a non-aborting Byzantine agreement protocol that is works when there are at least $2f_P+1$ processes with a decision in two delays in the common case. Other Abstract instances can be put between these two protocols to improve it further, but this does not concern us in this paper.
      \Naama{Added the last few sentences. Need to verify this. The reason I am worried is that I am not sure we actually show that either of them is really an algorithm for Abstract. In particular we don't show anything for the robust version. We just say it as a side comment.}
}
  

The algorithm has a special process $\ell$, say $\ell = p_1$, which serves both as a \emph{leader} and a \emph{follower}. Other processes act only as \emph{followers}.
The \disk{} is partitioned into $n+1$ regions denoted $\t{Region}[p]$ for each $p\in\Pi$, plus an extra one for $p_1$, $\t{Region}[\ell]$  in which it proposes a value.
Initially, $\t{Region}[p]$ is a regular SWMR region where $p$ is the writer.
Unlike in Algorithm~\ref{alg:non-eq}, some of the permissions are dynamic; processes may remove $p_1$'s write permission to $\t{Region}[\ell]$ (i.e., the \textit{legalChange} function returns false to any permission change requests, except for ones revoking $p_1$'s permission to write on $\t{Region}[\ell]$).

Processes initially execute under a \emph{normal} mode in common-case executions,
but may switch to \emph{panic} mode if they intend to abort, as in~\cite{aublin2015next}.
The pseudo-code of the normal mode is in Algorithm~\ref{alg:normal}.
$\t{Region}[p]$ contains three registers $\t{Value}[p]$, $\t{Panic}[p]$, $\t{Proof}[p]$ initially set to $\bot$, $\t{false}$, $\bot$.
To propose $v$, the leader $p_1$ signs $v$ and writes it to $\t{Value}[\ell]$.
If the write is successful (it may fail because its write permission was removed),
  then $p_1$ decides $v$; otherwise $p_1$ calls $\t{Panic\_mode}()$. Note that all processes, including $p_1$, continue their execution after deciding. However, $p_1$ never decides again if it decided as the leader.
A follower $q$ checks if $p_1$ wrote to $\t{Value}[\ell]$ and, if so, whether
  the value is properly signed.
If so, $q$ signs $v$, writes it to $\t{Value}[q]$, and waits for other processes to
  write the same value to $\t{Value}[*]$. If $q$ sees $2f+1$ copies of $v$ signed by different processes, $q$ assembles these copies in a \textit{unanimity proof}, which it signs and writes to $\t{Proof}[q]$. $q$ then waits for $2f+1$ unanimity proofs for $v$ to appear in $\t{Proof}[*]$, and checks that they are valid,
   in which case $q$ decides $v$.
This waiting continues until a timeout expires\footnote{The timeout is chosen to be an upper bound on the communication, processing and computation delays in the common case.}, 
  at which time $q$ calls $\t{Panic\_mode}()$.
In $\t{Panic\_mode}()$, a process $p$ sets $\t{Panic}[p]$ to $\t{true}$ to tell other processes
  it is panicking; other processes periodically check to see if they should panic too.
$p$ then removes write permission from $\t{Region}[\ell]$, and decides on a value to abort: either
  $\t{Value}[p]$ if it is non-$\bot$,  $\t{Value}[\ell]$ if it is non-$\bot$, or $p$'s input value.
  If  $p$ has a unanimity proof in $\t{Proof}[p]$, it adds it to the abort value. 

\ifcamera
In the full version of this paper~\cite{fullversion}, we prove the correctness of \fastbyz, and in particular we show the following two important agreement properties:
\else
In Appendix~\ref{sec:fastCorrectness}, we prove the correctness of \fastbyz, and in particular we show the following two important agreement properties:
\fi

\begin{lemma}[\fastbyz{} Decision Agreement]\label{lem:decisionAgreement}
	Let $p$ and $q$ be correct processes. If $p$ decides $v_1$ while $q$ decides $v_2$, then $v_1=v_2$.
\end{lemma}

\begin{lemma}[\fastbyz{} Abort Agreement]\label{lem:abortAgreement}
	Let $p$ and $q$ be correct processes (possibly identical). If $p$ decides $v$ in \fastbyz{} while $q$ aborts from \fastbyz{}, then
	$v$ will be $q$'s abort value. Furthermore, if $p$ is a follower, $q$'s abort proof is a correct unanimity proof. 
\end{lemma}

The above construction assumes a fail-free memory with regular
  registers, but we can extend it to tolerate memory failures
  using the approach of Section~\ref{sec:robustProtocol}, noting that
  each register has a single writer process.

\subsection{Putting it Together: the \composedAlgo{} Algorithm}\label{sec:composition}

\begin{figure}
\centering
\includegraphics[width=3in]{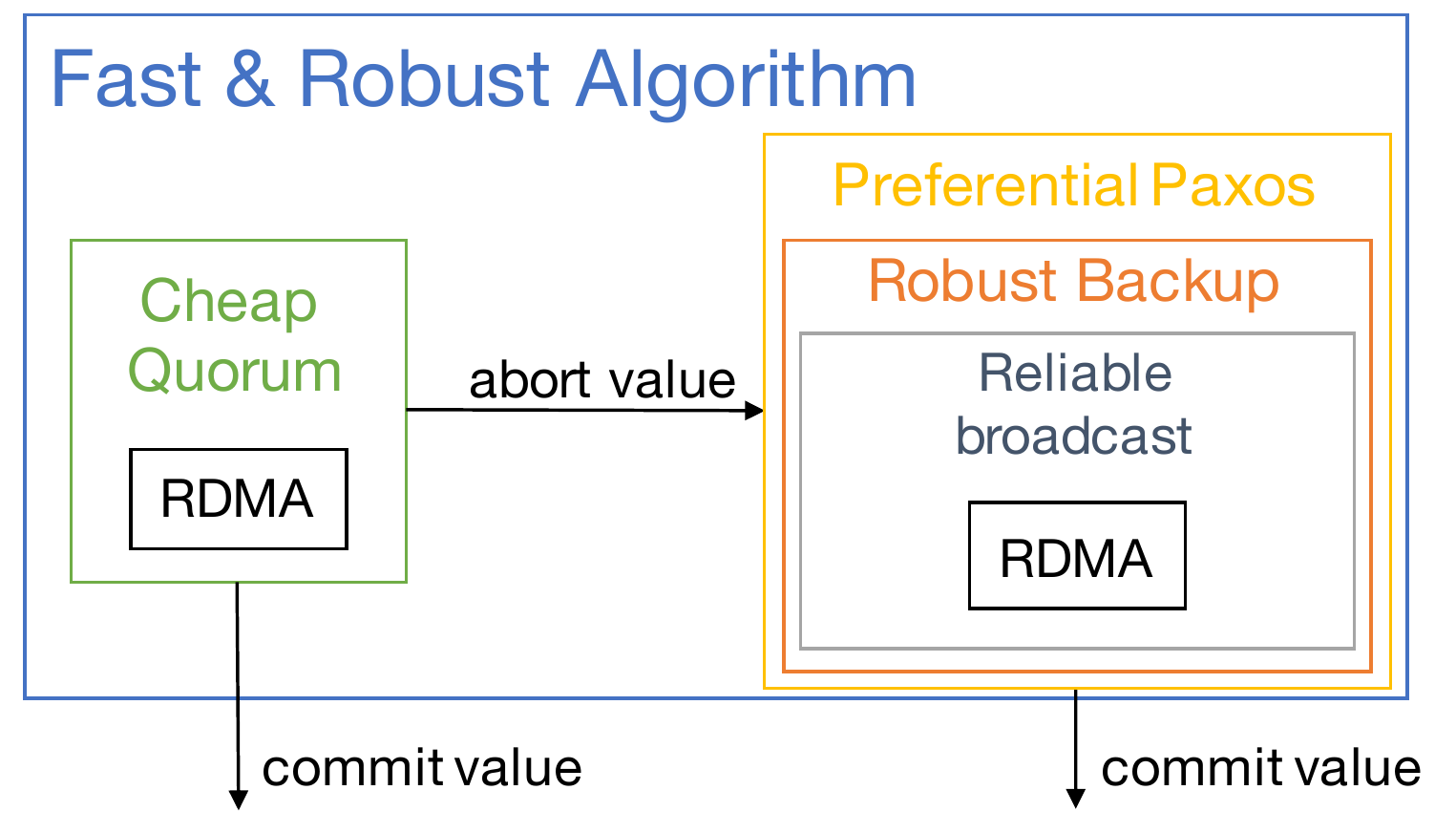}
\caption{Interactions of the components of the \composedAlgo{} Algorithm.}
\label{fig:byz_diagram}
\end{figure}

The final algorithm, called \composedAlgo{}, combines
  \fastbyz{} (\S\ref{sec:fastbyz}) and
  \robustbyz{} (\S\ref{sec:robustProtocol}), as we now explain.
Recall that \robustbyz{} 
  is parameterized by a message-passing consensus algorithm $\cal A$
  that tolerates crash-failures. $\cal A$ can be any such algorithm (e.g., Paxos).

Roughly, in \composedAlgo{}, we run \fastbyz{} and, if it aborts, 
  we use a process's
  abort value as its input value to \robustbyz{}.
However, we must carefully glue the two algorithms together  to ensure that
  if some correct process decided $v$ in \fastbyz{}, 
  then $v$ is the only value that can be decided in \robustbyz{}.

For this purpose, we propose a simple wrapper for \robustbyz{}, called \emph{\prefPax}. \prefPax{} first runs a set-up phase, in which processes may adopt new values,
and then runs \robustbyz{} with the new values. 
More specifically, there are some \emph{preferred} input values $v_1 \ldots v_k$, ordered by priority. We guarantee that every process adopts one of the top $f+1$ priority inputs. In particular, this means that if a majority of processes get the highest priority value, $v_1$, as input, then $v_1$ is guaranteed to be the decision value. 
The set-up phase is simple; all processes send each other their input values. Each process $p$ waits to receive $n-f$ such messages, and adopts the value with the highest priority that it sees. This is the value that $p$ uses as its input to Paxos.
\ifcamera
The pseudocode for \prefPax{} is given in the full version of our paper~\cite{fullversion}, where we also prove the following lemma about \prefPax:
\else
The pseudocode for \prefPax{} is given in Algorithm~\ref{alg:pref} in Appendix~\ref{sec:combinedCorrectness}, where we also prove the following lemma about \prefPax: 
\fi

\begin{lemma}[\prefPax{} Priority Decision]\label{lem:pref}
	\prefPax{} implements weak Byzantine agreement with $n \geq 2f_P+1$ processes. Furthermore, let $v_1, \ldots, v_n$ be the input values of an instance $C$ of \prefPax, ordered by priority. The decision value of correct processes is always one of $v_1, \ldots, v_{f+1}$. 
\end{lemma}

We can now describe \composedAlgo{} in detail.
We start executing \fastbyz{}. If \fastbyz{} aborts, we execute
  \prefPax{}, with each process receiving its abort value from \fastbyz{} as its input value to \prefPax{}. We define the priorities of inputs to \prefPax{} as follows.
  
  \begin{definition}[Input Priorities for \prefPax]\label{def:priorities}
  	The input values for \prefPax{} as it is used in \composedAlgo{} are split into three sets (here, $p_1$ is the leader of \fastbyz):
  	\begin{itemize}
  		\item $T = \{ v \mid v \text{ contains a correct unanimity proof }\}$
  		\item $M = \{v \mid v \not\in T \wedge v \text{ contains the signature of } p_1\}$
  		\item $B = \{v \mid v \not\in T \wedge v \not\in M\}$
  	\end{itemize}

  The priority order of the input values is such that for all values $v_T\in T$, $v_M \in M$, and $v_B \in B$, $priority(v_T) > priority(v_M) > priority(v_B)$.
  \end{definition} 
 Figure~\ref{fig:byz_diagram} shows how the various algorithms presented in this section come together to form the \composedAlgo{} algorithm.
\ifcamera
In the full version,
\else
In Appendices~\ref{sec:fastCorrectness} and~\ref{sec:combinedCorrectness}, 
\fi
we show that \composedAlgo{} is correct, with the following key lemma:

\begin{lemma}[Composition Lemma]\label{lem:composition}
	If some correct process decides a value $v$ in \fastbyz{} before an abort, then $v$ is the only value that can be decided in \prefPax{} with priorities as defined in Definition~\ref{def:priorities}. 
\end{lemma}

\begin{theorem}
There exists a 2-deciding algorithm for Weak Byzantine Agreement in a message-and-memory model with up to $f_P$ Byzantine processes and $f_M$ memory crashes,
  where $n \ge 2f_P+1$ and $m \ge 2f_M+1$.
\end{theorem}


	\section{Crash Failures} \label{sec:crash}

We now restrict ourselves to crash failures of processes and \disks{}.
Clearly, we can use the algorithms of Section~\ref{sec:byz} in this setting,
  to obtain a 2-deciding consensus algorithm
  with $n \ge 2f_P+1$ and $m \ge 2f_M+1$.
However, this is overkill since those algorithms use sophisticated mechanisms
  (signatures, non-equivocation) to guard against Byzantine behavior.
With only crash failures, we now show it is possible to retain the efficiency of
 a 2-deciding algorithm while improving
  resiliency.
In Section~\ref{sec:crashmajdisk},
  we first give a 2-deciding algorithm that allows the crash of all but one process
   ($n \ge f_P+1$) and a minority of \disks{} ($m \ge 2f_M+1$).
In Section~\ref{sec:crashmajcombined}, we
  improve resilience further by giving a 2-deciding algorithm that
  tolerates crashes of a minority of the combined set
  of \disks{} and processes.
 
\subsection{\crashAlgo{}} \label{sec:crashmajdisk}

Our starting point is the Disk Paxos algorithm~\cite{gafni2003disk}, which works in a system with
  processes and \disks{} where
  $n \ge f_P+1$ and $m \ge 2f_M+1$.
This is our resiliency goal, but Disk Paxos takes four delays in
  common executions.
Our new algorithm, called \crashAlgo{}, removes two delays; it retains the structure of Disk
  Paxos but uses permissions to skip steps.
Initially some fixed leader $\ell = p_1$ has exclusive write permission to all
  memories; if another process becomes leader, it takes the exclusive permission.
Having exclusive permission permits a leader $\ell$
  to optimize execution, because $\ell$
  can do two things simultaneously:
  (1) write its consensus proposal and (2) determine whether another 
  leader took over.
Specifically, if $\ell$ succeeds in (1), it knows no leader $\ell'$ took over because
  $\ell'$ would have taken the permission.
Thus $\ell$ avoids the last read in Disk Paxos, saving two delays.
Of course, care must be taken to implement this without violating safety.

\renewcommand{\figurename}{Algorithm}
\begin{figure}
    \caption{\crashAlgo---code for process $p$}
\begin{lstlisting}[keywords={},breaklines=true,tabsize=2]
Registers: for i=1..m, p @$\in \Pi$@,
    slot[i,p]: tuple (minProp, accProp, value)// in memory i
@$\Omega$@: failure detector that returns current leader

startPhase2(i) {
    add i to ListOfReady processes;
    while (size(ListOfReady)@$<$@majority of @\disks@) {}@\label{line:barrierPhase2}@
    Phase2Started = true;@\label{line:phase2start}@ }
    
propose(v) {
repeat forever {@\label{line:mainLoop}@
    wait until @$\Omega$@ == p; // wait to become leader @\label{line:leaderWait}@
    propNr = a higher value than any proposal number seen before;
    CurrentVal = v;@\label{line:currentValInit}@
    CurrentMaxProp = 0;
    Phase2Started = false;
    ListOfReady = @$\emptyset$@;
    for every @\disk@ i in parallel {@\label{line:pforStart}@
        if (p != p1 || not first attempt) {@\label{line:checkOmitPhase1}@
            getPermission(i);@\label{line:permission}@
            success = write(slot[i,p], (propNr,@$\bot$@,@$\bot$@));@\label{line:phase1write}@
            if (not success) { abort(); }@\label{line:restart2}@
            vals = read all slots from i; @\label{line:phase1read}@
            if (vals contains a non-null value) {
                val = v @$\in$@ vals with highest propNr; 
                if (val.propNr@$>$@propNr) { abort(); }@\label{line:Crestart1}@
                atomic {
                   if(val.propNr@$>$@CurrentMaxProp){
                        if (Phase2Started) {abort();}@\label{line:restartStraggler}@
                        CurrentVal = val.value;@\label{line:adoptValue}@
                        CurrentMaxProp = val.propNr;
                    } } }
            startPhase2(i);}@\label{line:endPhase1}@
        // done phase 1 or (p == p1 && p1's first attempt)
        success = write(slot[i,p], (propNr,propNr,CurrentVal)); @\label{line:phase2write}@
        if (not success) { abort(); } @\label{line:restart3}@
    } until this has been done at a majority of the @\disks@, or until 'abort' has been called@\label{line:pforEnd}@
    if (loop completed without abort) { 
        decide CurrentVal; } } }
        
\end{lstlisting}
\label{alg:consensus}
\end{figure}

The pseudocode of \crashAlgo{} is in Algorithm~\ref{alg:consensus}.
Each \disk{} has one \mregion{}, and at any time
  exactly one process can write to the region.
Each \disk{} $i$ holds a register $\t{slot}[i,p]$ for each process $p$.
Intuitively, $\t{slot}[i,p]$ is intended for $p$ to write, but $p$ may not have write
  permission to do that if it is not the leader---in that case, no process writes $\t{slot}[i,p]$.

When a process $p$ becomes the leader, it must execute a sequence of steps on a majority of the \disks{} to successfully commit a value. It is important that $p$ execute \emph{all} of these steps on each of the \disks{} that counts toward its majority; otherwise two leaders could miss each other's values and commit conflicting values. We therefore present the pseudocode for this algorithm in a \textit{parallel-for} loop (lines~\ref{line:pforStart}--\ref{line:pforEnd}), with one thread per \disk{} that $p$ accesses. The algorithm has two phases similar to Paxos, where the second phase may only begin after the first phase has been completed for a majority of the \disks. We represent this in the code with a barrier that waits for a majority of the threads.

When a process $p$ becomes leader, it executes the prepare phase  
  (the first leader $p_1$ can skip this phase in its first execution of the loop),
  where, for each \disk, $p$ attempts to (1) acquire exclusive write
  permission, (2) write a new proposal number in its slot, 
    and (3) read all slots of that \disk. $p$ waits to succeed in executing these steps on a majority of the \disks. 
If any of $p$'s writes fail 
  or $p$ finds a proposal with a higher proposal number, then $p$ gives up. This is represented with an \texttt{abort} in the pseudocode; when an \texttt{abort} is executed, the for loop terminates. We assume that when the for loop terminates---either because some thread has aborted or because a majority of threads have reached the end of the loop---all threads of the for loop are terminated and control returns to the main loop (lines~\ref{line:mainLoop}--\ref{line:pforEnd}).
  
  
  If $p$ does not abort, it adopts the value with highest proposal number of all those it read in the \disks. To make it clear  that  races should be avoided among parallel threads in the pseudocode, we wrap this part in an \texttt{atomic} environment.

In the next phase, each of $p$'s threads writes its value to its slot on its \disk{}.
If a write fails, $p$ gives up.
If $p$ succeeds, this is where we optimize time: $p$ can simply decide, whereas
  Disk Paxos must read the \disks{} again.
  
Note that it is possible that some of the \disks{} that made up the majority that passed the initial barrier may crash later on. To prevent $p$ from stalling forever in such a situation, it is important that straggler threads that complete phase 1 later on be allowed to participate in phase 2. However, if such a straggler thread observes a more up-to-date value in its \disk{} than the one adopted by $p$ for phase 2, this must be taken into account. In this case, to avoid inconsistencies, $p$ must abort its current attempt and restart the loop from scratch.

\rmv{  
    With permissions, we instead rely on knowing that if $p$'s writes succeeds, no other leaders could
      have written to the \disks{}.
    In the common case when $p_1$ is leader and there are no failures, $p_1$ decides quickly: it simply
      writes to a majority of \disks{} and decides.
      This incurs only two delays.
}

\rmv{
    Note that a process $p$ cannot start a new operation on a \disk{} $m$ before
      $p$'s previous operation on $m$ completed.
    That is, the second pfor loop for \disk{} $i$ should not start until the previous pfor iteration for $i$ has 
      finished.
    This is important for correctness; it avoids a situation in which $p$ writes its adopted value on a \disk{} that it did not read, and which could therefore have a proposal number higher than $p$'s.
}


The code ensures that some correct process eventually decides, but it is
  easy to extend it so all correct processes decide~\cite{chandra1996unreliable}, by having a decided process broadcast its decision.
Also, the code shows one instance of consensus, with $p_1$ as initial leader.
With many consensus instances, the leader terminates one instance and becomes
  the default leader in the next.

 
  \begin{theorem}\label{thm:crashAlgo}
  	Consider a message-and-memory model with up to $f_P$ process crashes and $f_M$ memory crashes,
  	where $n \ge f_P+1$ and $m \ge 2f_M+1$.
  	There exists a 2-deciding algorithm for consensus.
  \end{theorem}

\subsection{\combinedAlgo} \label{sec:crashmajcombined}

We now further enhance the failure resilience.
We show that \disks{} and processes are equivalent \emph{agents},
 in that it suffices for a majority of the agents (processes and \disks{} together) to remain alive to solve consensus.
Our new algorithm, \emph{\combinedAlgo}, achieves this resiliency. 
To do so, the algorithm relies on the ability to use both the messages and the \disks{} in
  our model; permissions are not needed.
%
%
The key idea is to align a message-passing algorithm and a \disk{}-based algorithm 
  to use any majority of agents. 
%
We align Paxos~\cite{lamport1998part} and \crashAlgo{} so that their decisions are coordinated.
More specifically, \crashAlgo{} and Paxos have two phases.
To align these algorithms, we factor out their differences and replace their steps with 
  an abstraction that is implemented differently for each algorithm.
The result is our \emph{\combinedAlgo} algorithm, which has two phases, each with three steps: \emph{communicate}, \emph{hear back}, and \emph{analyze}.
Each step
treats processes and \disks{} separately, and translates the results of operations on different agents to a common language. 
We implement the steps using their analogues in Paxos and \crashAlgo\footnote{We believe other implementations are possible. For example, replacing the \crashAlgo{} implementation for \disks{} with the Disk Paxos implementation yields an algorithm that does not use permissions.}.
\ifcamera
The pseudocode of \combinedAlgo{} is given in the full version of our paper~\cite{fullversion}.
\else
The pseudocode of \combinedAlgo{} is shown in Appendix~\ref{sec:alignedPseudo}.
\fi

	\section{Dynamic Permissions are Necessary for Efficient Consensus}\label{sec:imp}

In \S\ref{sec:crashmajdisk}, we showed how dynamic permissions can
  improve the performance of Disk Paxos.
Are dynamic permissions necessary?
We prove that with shared memory (or disks) alone,
  one cannot achieve 2-deciding consensus, even if the memory never fails,
  it has static permissions,
  processes may only fail by crashing, and the system is partially synchronous
  in the sense that eventually there is a known upper bound on the time it takes a
  correct process to take a step~\cite{dwork1988consensus}.
This result applies a fortiori to the Disk Paxos model~\cite{gafni2003disk}.

\begin{theorem}\label{thm:lower-bound}
   Consider a partially synchronous shared-memory model with registers, where registers
   can have arbitrary static permissions, memory never fails, and
   at most one processes may fail 
   by crashing.
   No consensus algorithm is 2-deciding.
\end{theorem}

\begin{proof}
	Assume by contradiction that $A$ is an algorithm in the stated model that is 2-deciding. That is, there is some execution $E$ of $A$ in which some process $p$ decides a value $v$ with 2 delays. 
	We denote by $R$ and $W$ the set of objects which $p$ reads and writes in $E$ respectively. Note that since $p$ decides in 2 delays in $E$, $R$ and $W$ must be disjoint, by the definition of operation delay and the fact that a process has at most one outstanding operation per object. Furthermore, $p$ must issue all of its read and writes without waiting for the response of any operation.
	
	%
	%
	Consider an execution $E'$ in which $p$ reads from the same set $R$ of objects and writes the same values as in $E$ to the same set $W$ of objects.
	All of the read operations that $p$ issues return by some time $t_0$, but the write operations of $p$ are delayed for a long time. Another process $p'$ begins its proposal of a value $v' \ne v$ after $t_0$. Since no process other than $p'$ writes to any objects, $E'$ is indistinguishable to $p'$ from an execution in which it runs alone. Since $A$ is a correct consensus algorithm that terminates if there is no contention, $p'$ must eventually decide value $v'$. Let $t'$ be the time at which $p'$ decides.
	All of $p$'s write operations terminate and are linearized in $E'$ after time $t'$. 
	Execution $E'$ is indistinguishable to $p$ from execution $E$, in which it ran alone. Therefore, $p$ decides $v \ne v'$, violating agreement.
	%
\end{proof}



    Theorem~\ref{thm:lower-bound}, together with the Fast Paxos algorithm of Lamport~\cite{lamport2006fast}, shows that an atomic read-write shared memory model is strictly weaker than the message passing model in its ability to solve consensus quickly. This result may be of independent interest, since often the classic shared memory and message passing models are seen as equivalent, because of the seminal computational equivalence result of Attiya, Bar-Noy, and Dolev~\cite{attiya1995sharing}. Interestingly, it is known that shared memory can tolerante more failures when solving consensus (with randomization or partial synchrony)~\cite{bracha1985asynchronous,aspnes1990fast}, and therefore it seems that perhaps shared memory is strictly stronger than message passing for solving consensus. However, our result shows that there are aspects in which message passing is stronger than shared memory. In particular, message passing can solve consensus faster than shared memory in well-behaved executions.

	\section{RDMA in Practice}\label{sec:implNotes}


Our model is meant to reflect capabilities of RDMA, while providing a clean abstraction to reason about.
We now give an overview of how RDMA works, and how features of our model can be implemented using RDMA.

RDMA enables a remote process to access local memory directly through the network interface card (NIC), without involving the CPU. For a piece of local memory to be accessible to a remote process $p$, the CPU has to \emph{register} that memory region and associate it with the appropriate connection (called \emph{Queue Pair}) for $p$.
The association of a registered memory region and a queue pair is done indirectly through a  \emph{protection domain}: both memory regions and queue pairs are associated with a protection domain, and a queue pair $q$ can be used to access a memory region $r$ if $q$ and $r$ and in the same protection domain.
The CPU must also specify what access level (read, write, read-write) is allowed to the memory region in each protection domain. 
A local memory area can thus be registered and associated with several queue pairs, with the same or different access levels, by associating it with one or more protection domains.
Each RDMA connection can be used by the remote
server to access registered memory regions using a unique
region-specific key created as a part of the registration process.

As highlighted by previous work~\cite{poke2015dare}, failures of the CPU, NIC and DRAM can be seen as independent (e.g., arbitrary delays, too many bit errors, failed ECC checks, respectively). For instance, \textit{zombie servers} in which the CPU is blocked but RDMA requests can still be served account for roughly half of all failures~\cite{poke2015dare}. This motivates our choice to treat processes and memory separately in our model. In practice, if a CPU fails permanently, the memory will also become unreachable through RDMA eventually; however, in such cases memory may remain available long enough for ongoing operations to complete. Also, in practical settings it is possible for full-system crashes to occur (e.g., machine restarts), which correspond to a process and a memory failing at the same time---this is allowed by our model.

Memory regions in our model correspond to RDMA memory regions. 
Static permissions can be implemented by making the appropriate memory region registration before the execution of the algorithm; these permissions then persist during execution without CPU involvement. Dynamic permissions require the host CPU to change the access levels; this should be done in the OS kernel: the kernel creates regions and controls their permissions, and then shares memory with user-space processes.
In this way, Byzantine processes cannot change permissions illegally. The assumption is that the kernel is not Byzantine. Alternatively, future hardware support similar to SGX could even allow parts of the kernel to be Byzantine.

Using RDMA, a process $p$ can grant permissions to a remote process $q$ by
registering memory regions with the appropriate access permissions
(read, write, or read/write) and sending the corresponding key to $q$.  
$p$ can revoke permissions dynamically by
simply deregistering the memory region.

For our reliable broadcast algorithm, each process can register the two
dimensional array of values in read-only mode with a protection
domain.  All the queue pairs used by that process are also created in the
context of the same protection domain.  Additionally, the process can
preserve write access permission to its row via another registration
of just that row with the protection domain, thus enabling
single-writer multiple-reader access.  Thereafter the reliable broadcast
algorithm can be implemented trivially by using RDMA reads and writes
by all processes.  Reliable broadcast with unreliable memories is
similarly straightforward since failure of the memory ensures that no
process will be able to access the memory.

For \fastbyz{}, the static memory region registrations are
straightforward as above.  To revoke the leader's write permission, it suffices for a region's host process to deregister the memory region.  Panic messages can be relayed using RDMA message
sends.

In our crash-only consensus algorithm, we leverage the capability of
registering overlapping memory regions in a protection domain.  As in
above algorithms, each process uses one protection domain for RDMA
accesses.  Queue pairs for connections to all other processes are
associated with this protection domain.  The process' entire slot
array is registered with the protection domain in read-only mode.  In
addition, the same slot array can be dynamically registered (and
deregistered) in write mode based on incoming write permission
requests: A proposer requests write permission using an RDMA message
send.  In response, the acceptor first deregisters write permission
for the immediate previous proposer.  The acceptor thereafter
registers the slot array in write mode and responds to the proposer
with the new key associated with the newly registered slot array.
Reads of the slot array are performed by the proposer using RDMA
reads.  Subsequent second phase RDMA write of the value can be
performed on the slot array as long as the proposer continues to have
write permission to the slot array.  The RDMA write fails if the
acceptor granted write permission to another proposer in the meantime.


    \paragraph{\bf Acknowledgements}
	We wish to thank the anonymous reviewers for their helpful comments on improving the paper. This work has been supported in
    part by the European Research Council (ERC) Grant 339539 (AOC), and a Microsoft PhD Fellowship.
	
	\bibliographystyle{plain}
	\bibliography{biblio}

\begin{thebibliography}{10}

\bibitem{abraham2006Byzantine}
Ittai Abraham, Gregory Chockler, Idit Keidar, and Dahlia Malkhi.
\newblock Byzantine disk paxos: optimal resilience with byzantine shared
  memory.
\newblock {\em Distributed computing (DIST)}, 18(5):387--408, 2006.

\bibitem{AGMT1992}
Yehuda Afek, David~S. Greenberg, Michael Merritt, and Gadi Taubenfeld.
\newblock Computing with faulty shared memory.
\newblock In {\em ACM Symposium on Principles of Distributed Computing (PODC)},
  pages 47--58, August 1992.

\bibitem{aguilera2018passing}
Marcos~K Aguilera, Naama Ben-David, Irina Calciu, Rachid Guerraoui, Erez
  Petrank, and Sam Toueg.
\newblock Passing messages while sharing memory.
\newblock In {\em ACM Symposium on Principles of Distributed Computing (PODC)},
  pages 51--60. ACM, 2018.

\bibitem{alon2005tight}
Noga Alon, Michael Merritt, Omer Reingold, Gadi Taubenfeld, and Rebecca~N
  Wright.
\newblock Tight bounds for shared memory systems accessed by byzantine
  processes.
\newblock {\em Distributed computing (DIST)}, 18(2):99--109, 2005.

\bibitem{aspnes1990fast}
James Aspnes and Maurice Herlihy.
\newblock Fast randomized consensus using shared memory.
\newblock {\em Journal of algorithms}, 11(3):441--461, 1990.

\bibitem{attiya1995sharing}
Hagit Attiya, Amotz Bar-Noy, and Danny Dolev.
\newblock Sharing memory robustly in message-passing systems.
\newblock {\em Journal of the ACM (JACM)}, 42(1):124--142, 1995.

\bibitem{aublin2015next}
Pierre-Louis Aublin, Rachid Guerraoui, Nikola Kne{\v{z}}evi{\'c}, Vivien
  Qu{\'e}ma, and Marko Vukoli{\'c}.
\newblock The next 700 {BFT} protocols.
\newblock {\em ACM Transactions on Computer Systems (TOCS)}, 32(4):12, 2015.

\bibitem{bazzi1991optimally}
Rida Bazzi and Gil Neiger.
\newblock Optimally simulating crash failures in a byzantine environment.
\newblock In {\em International Workshop on Distributed Algorithms (WDAG)},
  pages 108--128. Springer, 1991.

\bibitem{behrens2016derecho}
Jonathan Behrens, Ken Birman, Sagar Jha, Matthew Milano, Edward Tremel, Eugene
  Bagdasaryan, Theo Gkountouvas, Weijia Song, and Robbert Van~Renesse.
\newblock Derecho: Group communication at the speed of light.
\newblock Technical report, Technical Report. Cornell University, 2016.

\bibitem{ben1983another}
Michael Ben-Or.
\newblock Another advantage of free choice (extended abstract): Completely
  asynchronous agreement protocols.
\newblock In {\em ACM Symposium on Principles of Distributed Computing (PODC)},
  pages 27--30. ACM, 1983.

\bibitem{bessani2009sharing}
Alysson~Neves Bessani, Miguel Correia, Joni da~Silva~Fraga, and Lau~Cheuk Lung.
\newblock Sharing memory between byzantine processes using policy-enforced
  tuple spaces.
\newblock {\em IEEE Transactions on Parallel and Distributed Systems},
  20(3):419--432, 2009.

\bibitem{fastpaxos}
Romain Boichat, Partha Dutta, Svend Frolund, and Rachid Guerraoui.
\newblock Reconstructing paxos.
\newblock {\em SIGACT News}, 34(2):42--57, March 2003.

\bibitem{bouzid2016necessary}
Zohir Bouzid, Damien Imbs, and Michel Raynal.
\newblock A necessary condition for byzantine k-set agreement.
\newblock {\em Information Processing Letters}, 116(12):757--759, 2016.

\bibitem{bracha1987asynchronous}
Gabriel Bracha.
\newblock Asynchronous byzantine agreement protocols.
\newblock {\em Information and Computation}, 75(2):130--143, 1987.

\bibitem{bracha1985asynchronous}
Gabriel Bracha and Sam Toueg.
\newblock Asynchronous consensus and broadcast protocols.
\newblock {\em Journal of the ACM (JACM)}, 32(4):824--840, 1985.

\bibitem{brasileiro2001consensus}
Francisco Brasileiro, Fab{\'\i}ola Greve, Achour Most{\'e}faoui, and Michel
  Raynal.
\newblock Consensus in one communication step.
\newblock In {\em International Conference on Parallel Computing Technologies},
  pages 42--50. Springer, 2001.

\bibitem{RachidBook}
Christian Cachin, Rachid Guerraoui, and Lu{\'{\i}}s E.~T. Rodrigues.
\newblock {\em Introduction to Reliable and Secure Distributed Programming
  {(2.} ed.)}.
\newblock Springer, 2011.

\bibitem{chandra1996unreliable}
Tushar~Deepak Chandra and Sam Toueg.
\newblock Unreliable failure detectors for reliable distributed systems.
\newblock {\em Journal of the ACM (JACM)}, 43(2):225--267, 1996.

\bibitem{ChunMS08}
Byung{-}Gon Chun, Petros Maniatis, and Scott Shenker.
\newblock Diverse replication for single-machine byzantine-fault tolerance.
\newblock In {\em USENIX Annual Technical Conference (ATC)}, pages 287--292,
  2008.

\bibitem{ChunMSK07}
Byung{-}Gon Chun, Petros Maniatis, Scott Shenker, and John Kubiatowicz.
\newblock Attested append-only memory: making adversaries stick to their word.
\newblock In {\em ACM Symposium on Operating Systems Principles (SOSP)}, pages
  189--204, 2007.

\bibitem{clement2012limited}
Allen Clement, Flavio Junqueira, Aniket Kate, and Rodrigo Rodrigues.
\newblock On the (limited) power of non-equivocation.
\newblock In {\em ACM Symposium on Principles of Distributed Computing (PODC)},
  pages 301--308. ACM, 2012.

\bibitem{CorreiaNV04}
Miguel Correia, Nuno~Ferreira Neves, and Paulo Ver{\'{\i}}ssimo.
\newblock How to tolerate half less one byzantine nodes in practical
  distributed systems.
\newblock In {\em International Symposium on Reliable Distributed Systems
  (SRDS)}, pages 174--183, 2004.

\bibitem{correia2010asynchronous}
Miguel Correia, Giuliana~S Veronese, and Lau~Cheuk Lung.
\newblock Asynchronous byzantine consensus with 2f+ 1 processes.
\newblock In {\em ACM symposium on applied computing (SAC)}, pages 475--480.
  ACM, 2010.

\bibitem{dobre2006one}
Dan Dobre and Neeraj Suri.
\newblock One-step consensus with zero-degradation.
\newblock In {\em International Conference on Dependable Systems and Networks
  (DSN)}, pages 137--146. IEEE Computer Society, 2006.

\bibitem{dragojevic2014farm}
Aleksandar Dragojevi{\'c}, Dushyanth Narayanan, Miguel Castro, and Orion
  Hodson.
\newblock {FaRM}: Fast remote memory.
\newblock In {\em USENIX Symposium on Networked Systems Design and
  Implementation (NSDI)}, pages 401--414, 2014.

\bibitem{dutta2005fast}
Partha Dutta, Rachid Guerraoui, and Leslie Lamport.
\newblock How fast can eventual synchrony lead to consensus?
\newblock In {\em International Conference on Dependable Systems and Networks
  (DSN)}, pages 22--27. IEEE, 2005.

\bibitem{dwork1988consensus}
Cynthia Dwork, Nancy Lynch, and Larry Stockmeyer.
\newblock Consensus in the presence of partial synchrony.
\newblock {\em Journal of the ACM (JACM)}, 35(2):288--323, 1988.

\bibitem{fischer1985impossibility}
Michael~J Fischer, Nancy~A Lynch, and Michael~S Paterson.
\newblock Impossibility of distributed consensus with one faulty process.
\newblock {\em Journal of the ACM (JACM)}, 1985.

\bibitem{gafni2003disk}
Eli Gafni and Leslie Lamport.
\newblock Disk paxos.
\newblock {\em Distributed computing (DIST)}, 16(1):1--20, 2003.

\bibitem{JCT1998}
Prasad Jayanti, Tushar~Deepak Chandra, and Sam Toueg.
\newblock Fault-tolerant wait-free shared objects.
\newblock {\em Journal of the ACM (JACM)}, 45(3):451--500, May 1998.

\bibitem{kalia2015using}
Anuj Kalia, Michael Kaminsky, and David~G Andersen.
\newblock Using {RDMA} efficiently for key-value services.
\newblock {\em ACM SIGCOMM Computer Communication Review}, 44(4):295--306,
  2015.

\bibitem{kaminsky2016design}
Anuj Kalia, Michael Kaminsky, and David~G Andersen.
\newblock Design guidelines for high performance {RDMA} systems.
\newblock In {\em USENIX Annual Technical Conference (ATC)}, page 437, 2016.

\bibitem{kalia2016fasst}
Anuj Kalia, Michael Kaminsky, and David~G Andersen.
\newblock {FaSST}: Fast, scalable and simple distributed transactions with
  two-sided {(RDMA)} datagram {RPC}s.
\newblock In {\em USENIX Symposium on Operating System Design and
  Implementation (OSDI)}, volume~16, pages 185--201, 2016.

\bibitem{KapitzaBCDKMSS12}
R{\"{u}}diger Kapitza, Johannes Behl, Christian Cachin, Tobias Distler, Simon
  Kuhnle, Seyed~Vahid Mohammadi, Wolfgang Schr{\"{o}}der{-}Preikschat, and
  Klaus Stengel.
\newblock Cheapbft: resource-efficient byzantine fault tolerance.
\newblock In {\em European Conference on Computer Systems (EuroSys)}, pages
  295--308, 2012.

\bibitem{keidar2001cost}
Idit Keidar and Sergio Rajsbaum.
\newblock On the cost of fault-tolerant consensus when there are no faults:
  preliminary version.
\newblock {\em ACM SIGACT News}, 32(2):45--63, 2001.

\bibitem{kursawe2002}
Klaus Kursawe.
\newblock Optimistic byzantine agreement.
\newblock In {\em International Symposium on Reliable Distributed Systems
  (SRDS)}, pages 262--267, October 2002.

\bibitem{Lam83}
Leslie Lamport.
\newblock The weak byzantine generals problem.
\newblock {\em Journal of the ACM (JACM)}, 30(3):668--676, July 1983.

\bibitem{lamport1998part}
Leslie Lamport.
\newblock The part-time parliament.
\newblock {\em ACM Transactions on Computer Systems (TOCS)}, 16(2):133--169,
  1998.

\bibitem{lamport2006fast}
Leslie Lamport.
\newblock Fast paxos.
\newblock {\em Distributed computing (DIST)}, 19(2):79--103, 2006.

\bibitem{lamport1982Byzantine}
Leslie Lamport, Robert Shostak, and Marshall Pease.
\newblock The byzantine generals problem.
\newblock {\em ACM Transactions on Programming Languages and Systems (TOPLAS)},
  4(3):382--401, 1982.

\bibitem{malkhi2003objects}
Dahlia Malkhi, Michael Merritt, Michael~K Reiter, and Gadi Taubenfeld.
\newblock Objects shared by byzantine processes.
\newblock {\em Distributed computing (DIST)}, 16(1):37--48, 2003.

\bibitem{MA2006}
Jean-Philippe Martin and Lorenzo Alvisi.
\newblock Fast byzantine consensus.
\newblock {\em IEEE Transactions on Dependable and Secure Computing (TDSC)},
  3(3):202--215, July 2006.

\bibitem{neiger1990automatically}
Gil Neiger and Sam Toueg.
\newblock Automatically increasing the fault-tolerance of distributed
  algorithms.
\newblock {\em Journal of Algorithms}, 11(3):374--419, 1990.

\bibitem{pease1980reaching}
Marshall Pease, Robert Shostak, and Leslie Lamport.
\newblock Reaching agreement in the presence of faults.
\newblock {\em Journal of the ACM (JACM)}, 27(2):228--234, 1980.

\bibitem{poke2015dare}
Marius Poke and Torsten Hoefler.
\newblock {DARE}: High-performance state machine replication on {RDMA}
  networks.
\newblock In {\em Symposium on High-Performance Parallel and Distributed
  Computing (HPDC)}, pages 107--118. ACM, 2015.

\bibitem{rusch2018towards}
Signe R{\"u}sch, Ines Messadi, and R{\"u}diger Kapitza.
\newblock Towards low-latency byzantine agreement protocols using {RDMA}.
\newblock In {\em IEEE/IFIP International Conference on Dependable Systems and
  Networks Workshops (DSN-W)}, pages 146--151. IEEE, 2018.

\bibitem{bosco}
Yee~Jiun Song and Robbert van Renesse.
\newblock Bosco: One-step byzantine asynchronous consensus.
\newblock In {\em International Symposium on Distributed Computing (DISC)},
  pages 438--450, September 2008.

\bibitem{VeroneseCBLV13}
Giuliana~Santos Veronese, Miguel Correia, Alysson~Neves Bessani, Lau~Cheuk
  Lung, and Paulo Ver{\'{\i}}ssimo.
\newblock Efficient byzantine fault-tolerance.
\newblock {\em {IEEE} Trans. Computers}, 62(1):16--30, 2013.

\bibitem{wang2017apus}
Cheng Wang, Jianyu Jiang, Xusheng Chen, Ning Yi, and Heming Cui.
\newblock {APUS}: Fast and scalable paxos on {RDMA}.
\newblock In {\em Symposium on Cloud Computing (SoCC)}, pages 94--107. ACM,
  2017.

\end{thebibliography}
	
	\ifcamera
	\else
	\appendix
	\section{Correctness of Reliable Broadcast}\label{sec:nonEquivCorrectness}

\begin{observation}\label{obs:once}
	In Algorithm~\ref{alg:non-eq}, if $p$ is a correct process, then no slot that belongs to $p$ is written to more than once.
\end{observation}

\begin{proof}
		Since $p$ is correct, $p$ never writes on any slot more than once. Furthermore, since all slots are single-writer registers, no other process can write on these slots. 
\end{proof}

\begin{observation}\label{obs:infinitelyOften}
    In Algorithm~\ref{alg:non-eq}, correct processes invoke and return from try\_deliver(q) infinitely often, for all $q\in\Pi$.
\end{observation}
\begin{proof}
    The try\_deliver() function does not contain any blocking steps, loops or goto statements. Thus, if a correct process invokes try\_deliver(), it will eventually return. Therefore, for a fixed $q$ the infinite loop at line~\ref{line:tryDeliver} will invoke and return try\_deliver($q$) infinitely often. Since the parallel for loop at line~\ref{line:forloop} performs the infinite loop in parallel for each $q\in\Pi$, the Observation holds.
\end{proof}

\begin{proof}[Proof of Lemma~\ref{lem:SMnon-eq}]
	We prove the lemma by showing that Algorithm~\ref{alg:non-eq} correctly implements reliable broadcast. 
	That is, we need to show that Algorithm~\ref{alg:non-eq} satisfies the four properties of reliable broadcast.
	
	\textbf{Property 1.}
Let $p$ be a correct process that broadcasts $(k,m)$. We show that all correct processes eventually deliver $(k,m)$ from $p$. Assume by contradiction that there exists some correct process $q$ which does not deliver $(k,m)$. Furthermore, assume without loss of generality that $k$ is the smallest key for which Property 1 is broken. That is, all correct processes must eventually deliver all messages $(k',m')$ from $p$, for $k' < k$. Thus, all correct processes must eventually increment $last[p]$ to $k$.

We consider two cases, depending on whether or not some process eventually writes an L2 proof for some $(k,m')$ message from $p$ in its L2Proof slot. 

First consider the case where no process ever writes an L2 proof of any value $(k,m')$ from $p$.
Since $p$ is correct, upon broadcasting $(k,m)$, $p$ must sign and write $(k,m)$ into $Values[p,k,p]$ at line~\ref{line:broadcastWrite}. By Observation~\ref{obs:once}, $(k,m)$ will remain in that slot forever. Because of this, and because there is no L2 proof, all correct processes, after reaching $last[p] = k$, will eventually read  $(k,m)$ in line~\ref{line:readMessage}, write it into their own Value slot in line~\ref{line:copyVal} and change their state to WaitForL1Proof. 

Furthermore, since $p$ is correct and we assume signatures are unforgeable, no process $q$ can write any other valid value $(k',m')\neq(k,m)$ into its $Values[q,k,p]$ slot. Thus, eventually each correct process $q$ will add at least $f+1$ copies of $(k,m)$ to its checkedVals, write an L1proof consisting of these values into $L1Proof[q,k,p]$ in line~\ref{line:writel1prf}, and change their state to WaitForL2Proof. 

Therefore, all correct processes will eventually read at least $f+1$ valid L1 Proofs for $(k,m)$ in line~\ref{line:readL1prf} and construct and write valid L2 proofs for $(k,m)$. This contradicts the assumption that no L2 proof ever gets written. 

In the case where there is some L2 proof, by the argument above, the only value it can prove is $(k,m)$. 
Therefore, all correct processes will see at least one valid L2 proof at deliver. This contradicts our assumption that $q$ is correct but does not deliver $(k,m)$ from $p$.

\textbf{Property 2.}
We now prove the second property of reliable broadcast.
Let $p$ and $p'$ be any two correct processes, and $q$ be some process, such that $p$ delivers $(k,m)$ from $q$ and $p'$ delivers $(k,m')$ from $q$.
Assume by contradiction that $m\neq m'$.
	
Since $p$ and $p'$ are correct, they must have seen valid L2 proofs at line~\ref{line:readL2proof} before delivering $(k,m)$ and $(k,m')$ respectively. Let $\mathcal{P}$ and $\mathcal{P'}$ be those valid proofs for $(k,m)$ and $(k,m')$ respectively. $\mathcal{P}$ (resp. $\mathcal{P'}$) consists of at least $f+1$ valid L1 proofs; therefore, at least one of those proofs was created by some correct process $r$ (resp. $r'$). Since $r$ (resp. $r'$) is correct, it must have written $(k,m)$ (resp. $(k,m')$) to its Values slot in line~\ref{line:copyVal}. Note that after copying a value to their slot, in the WaitForL1Proof state, correct processes read \emph{all} Value slots line~\ref{line:readValCopies}. Thus, both $r$ and $r'$ read all Value slots  before compiling their L1 proof for $(k,m)$ (resp. $(k,m')$). 

Assume without loss of generality that $r$ wrote $(k,m)$ before $r'$ wrote $(k,m')$; by Observation~\ref{obs:once}, it must then be the case that $r'$ later saw both $(k,m)$ and $(k,m')$ when it read all Values slots (line~\ref{line:readValCopies}). Since $r'$ is correct, it cannot have then compiled an L1 proof for $(k,m')$ (the check at line~\ref{line:checkL1proof} failed). We have reached a contradiction.
	
\textbf{Property 3.} 
We show that if a correct process $p$ delivers $(k,m)$ from a correct process $p'$, then $p'$ broadcast $(k,m)$.
Correct processes only deliver values for which a valid L2 proof exists (lines~\ref{line:readL2proof}---\ref{line:deliver}). Therefore, $p$ must have seen a valid L2 proof $\mathcal{P}$ for $(k,m)$. $\mathcal{P}$ consists of at least $f+1$ L1 proofs for $(k,m)$ and each L1 proof consists of at least $f+1$ matching copies of $(k,m)$, signed by $p'$. Since $p'$ is correct and we assume signatures are unforgeable, $p'$ must have broadcast $(k,m)$ (otherwise $p'$ would not have attached its signature to $(k,m)$). 

\textbf{Property 4.}
Let $p$ be a correct process such that $p$ delivers $(k,m)$ from $q$. We show that all correct process must deliver $(k,m')$ from $q$, for some $m'$.

By construction of the algorithm, if $p$ delivers $(k,m)$ from $q$, then for all $i<k$ there exists $m_i$ such that $p$ delivered $(i,m_i)$ from $q$ before delivering $(k,m)$ (this is because $p$ can only deliver $(k,m)$ if $last[q] = k$ and $last[q]$ is only incremented to $k$ after $p$ delivers $(k-1,m_{k-1})$). 

Assume by contradiction that there exists some correct process $r$ which does not deliver $(k,m')$ from $q$, for any $m'$. Further assume without loss of generality that $k$ is the smallest key for which $r$ does not deliver any message from $q$. Thus, $r$ must have delivered $(i,m_i')$ from $q$ for all $i<k$; thus, $r$ must have incremented $last[q]$ to $k$. Since $r$ never delivers any message from $q$ for key $k$, $r$'s $last[q]$ will never increase past $k$.

Since $p$ delivers $(k,m)$ from $q$, then $p$  must have written a valid L2 proof $\mathcal{P}$ of $(k,m)$ in its L2Proof slot in line~\ref{line:copyL2proof}. By Observation~\ref{obs:once}, $\mathcal{P}$ will remain in $p$'s L2Proof[p,k,q] slot forever. Thus, at least one of the slots $L2Proof[\cdot,k,q]$ will forever contain a valid L2 proof. Since $r$'s $last[q]$ eventually reaches $k$ and never increases past $k$, $r$ will eventually (by Observation~\ref{obs:infinitelyOften}) see a valid L2 proof in line~\ref{line:readL2proof} and deliver a message for key $k$ from $q$. We have reached a contradiction.
\end{proof}
	\section{Correctness of \fastbyz{}}\label{sec:fastCorrectness}

We prove that \fastbyz{} satisfies certain useful properties that will help us show that it composes with \prefPax{} to form a correct weak Byzantine agreement protocol.
For the proofs, we first formalize some terminology.
We say that a process \emph{proposed} a value $v$ by time $t$ if it successfully executes line~\ref{line:propose}; that is, $p$ receives the response $ack$ in line~\ref{line:propose} by $t$.
When a process aborts, note that it outputs a tuple. We say that the first element of its tuple is its \emph{abort value}, and the second is its \emph{abort proof}. We sometimes say that a process $p$ aborts with value $v$ and proof $pr$, meaning that $p$ outputs $(v, pr)$ in its abort.
Furthermore, the value in a process $p$'s Proof region is called a \emph{correct unanimity proof} if it contains $n$ copies of the same value, each correctly signed by a different process.

\begin{observation}\label{obs:fastonce}
	In \fastbyz{}, no value written by a correct process is ever overwritten.
\end{observation}

\begin{proof}
	By inspecting the code, we can see that the correct behavior is for processes to never overwrite any values. Furthermore, since all regions are initially single-writer, and the legalChange function never allows another process to acquire write permission on a region that they cannot write to initially, no other process can overwrite these values.
\end{proof}

\begin{lemma}[\fastbyz{} Validity]\label{lem:fastValidity}
    In \fastbyz{}, if there are no faulty processes and some process decides $v$, then $v$ is the input of some process.
\end{lemma}
\begin{proof}
If $p=p_1$, the lemma is trivially true, because $p_1$ can only decide on its input value. If $p\neq p_1$, $p$ can only decide on a value $v$ if it read that value from the leader's region. Since only the leader can write to its region, it follows that $p$ can only decide on a value that was proposed by the leader ($p_1$).
\end{proof}

\begin{lemma}[\fastbyz{} Termination]\label{lem:fastTermination}
    If a correct process $p$ proposes  some value, every correct process $q$ will decide a value or abort.
\end{lemma}
\begin{proof}
    Clearly, if $q = p_1$ proposes a value, then $q$ decides. Now let $q \neq p_1$ be a correct follower and assume $p_1$ is a correct leader that proposes $v$.
Since $p_1$ proposed $v$, $p_1$ was able to write $v$ in the leader region, where $v$ remains forever by Observation~\ref{obs:fastonce}.
Clearly, if $q$ eventually enters panic mode, then it eventually aborts; there is no waiting done in panic mode. 
If $q$ never enters panic mode, then $q$ eventually sees $v$ on the leader region and eventually finds $2f+1$ copies of $v$ on the regions of other followers (otherwise $q$ would enter panic mode). Thus $q$ eventually decides $v$.
\end{proof}

\begin{lemma}[\fastbyz{} Progress]\label{lem:progress}
    If the system is synchronous and all processes are correct, then no correct process aborts in \fastbyz{}.
\end{lemma}
\begin{proof}
    Assume the contrary: there exists an execution in which the system is synchronous and all processes are correct, yet some process aborts. Processes can only abort after entering panic mode, so let $t$ be the first time when a process enters panic mode and let $p$ be that process. Since $p$ cannot have seen any other process declare panic, $p$ must have either timed out at line~\ref{line:timeout1} or~\ref{line:timeout2}, or its checks failed on line~\ref{line:checkLeader}. However, since the entire system is synchronous and $p$ is correct, $p$ could not have panicked because of a time-out at line~\ref{line:timeout1}. So, $p_1$ must have written its value $v$, correctly signed, to $p_1$'s region at a time $t' < t$. Therefore, $p$ also could not have panicked by failing its checks on line~\ref{line:checkLeader}.
    Finally, since all processes are correct and the system is synchronous, all processes must have seen $p_1$'s value and copied it to their slot. Thus, $p$ must have seen these values and decided on $v$ at line~\ref{line:countCheck}, contradicting the assumption that $p$ entered panic mode. 
\end{proof}

\begin{lemma}[Lemma~\ref{lem:decisionAgreement}: \fastbyz{} Decision Agreement]
    Let $p$ and $q$ be correct processes. If $p$ decides $v_1$ while $q$ decides $v_2$, then $v_1=v_2$.
\end{lemma}
\begin{proof}
    Assume the property does not hold:  $p$ decided some value $v_1$ and $q$ decided some different value $v_2$. Since $p$ decided $v_1$, then $p$ must have seen a copy of $v_1$ at $2f_P+1$ replicas, including $q$. But then $q$ cannot have decided $v_2$, because by Observation~\ref{obs:fastonce}, $v_1$ never gets overwritten from $q$'s region, and by  the code, $q$ only can decide a value written in its region.
\end{proof}

\begin{lemma}[Lemma \ref{lem:abortAgreement}: \fastbyz{} Abort Agreement]
	Let $p$ and $q$ be correct processes (possibly identical). If $p$ decides $v$ in \fastbyz{} while $q$ aborts from \fastbyz{}, then
    $v$ will be $q$'s abort value. Furthermore, if $p$ is a follower, $q$'s abort proof is a correct unanimity proof. 
\end{lemma}
\begin{proof}
   If $p=q$, the property follows immediately, because of lines~\ref{line:readOwn} through~\ref{line:myval} of panic mode.
If $p\neq q$, we consider two cases:
\begin{itemize}
    \item If $p$ is a follower, then for $p$ to decide, all processes, and in particular, $q$, must have replicated both $v$ and a correct proof of unanimity before $p$ decided. 
    Therefore, by Observation~\ref{obs:fastonce}, $v$ and the unanimity proof are still there when $q$ executes the panic code in lines~\ref{line:readOwn} through~\ref{line:myval}.
    Therefore $q$ will abort with $v$ as its value and a correct unanimity proof as its abort proof. 
    \item If $p$ is the leader, then first note that since $p$ is correct, by Observation~\ref{obs:fastonce} $v$ remains the value written in the leader's Value region. There are two cases. If $q$ has replicated a value into its Value region, then it must have read it from $Value[p_1]$, and therefore it must be $v$. Again by Observation~\ref{obs:fastonce}, $v$ must still be the value written in $q$'s Value region when $q$ executes the panic code. Therefore $q$ aborts with value $v$.
    Otherwise, if $q$ has not replicated a value, then $q$'s Value region must be empty at the time of the panic, since the legalChange function disallows other processes from writing on that region. Therefore $q$ reads $v$ from $Value[p_1]$ and aborts with $v$.  \qedhere
\end{itemize} 
\end{proof}

\begin{lemma}
    \fastbyz{} is 2-deciding.
\end{lemma}
\begin{proof}
    Consider an execution in which every process is correct and the system is synchronous. Then no process will enter panic mode (by Lemma~\ref{lem:progress}) and thus $p_1$ will not have its permission revoked. $p_1$ will therefore be able to write its input value to $p_1$'s region and decide after this single write (2 delays). 
\end{proof}

	\section{Correctness of the \composedAlgo{}}\label{sec:combinedCorrectness}

 The following is the pseudocode of \prefPax{}. Recall that \emph{T-send} and \emph{T-receive} are the trusted message passing primitives that are implemented in \cite{clement2012limited} using non-equivocation and signatures.
\renewcommand{\figurename}{Algorithm}
\begin{figure}[h!]
	\caption{\prefPax{}---code for process $p$}
	\begin{lstlisting}[columns=fullflexible,breaklines=true, keywords={}]
	propose((v, priorityTag)){
    	T-send(v, priorityTag) to all;
    	Wait to T-receive (val,priorityTag) from @$n-f_P$@ processes;
    	best = value with highest priority out of messages received;
    	RobustBackup(Paxos).propose(best);	} @\label{line:paxosCall}
	\end{lstlisting}
	\label{alg:pref}
\end{figure}
\renewcommand{\figurename}{Figure}

\begin{lemma}[Lemma \ref{lem:pref}: \prefPax{} Priority Decision]
	\prefPax{} implements weak Byzantine agreement with $n \geq 2f_P+1$ processes. Furthermore, let $v_1, \ldots, v_n$ be the input values of an instance $C$ of \prefPax, ordered by priority. The decision value of correct processes is always one of $v_1, \ldots, v_{f_P+1}$. 
\end{lemma}

\begin{proof}
	By Lemma~\ref{lem:robustByz}, \robustbyz(Paxos) solves weak Byzantine agreement with $n \geq 2f_P+1$ processes.
%
	Note that before calling \robustbyz(Paxos), each process may change its input, but only to the input of another process. Thus, by the correctness and fault tolerance of Paxos, \prefPax{} clearly solves weak Byzantine agreement with $n\geq 2f_P+1$ processes.
		Thus we only need to show that \prefPax{} 
	satisfies the priority decision property with $2f_P+1$ processes that may only fail by crashing. 

	Since \robustbyz(Paxos) satisfies validity, if all processes call \robustbyz(Paxos) in line~\ref{line:paxosCall} with a value $v$ that is one of the $f_P+1$ top priority values (that is, $v \in \{v_1,\ldots,v_{f_P+1}\}$), then the decision of correct processes will also be in $\{v_1,\ldots,v_{f_P+1}\}$.
	So we just need to show that every process indeed adopts one of the top $f_P+1$ values. Note that each process $p$ waits to see $n-f_P$ values, and then picks the highest priority value that it saw. No process can lie or pick a different value, since we use T-send and T-receive throughout. Thus, $p$ can miss at most $f_P$ values that are higher priority than the one that it adopts.
\end{proof}

We now prove the following key composition property that shows that the composition of \fastbyz{} and \prefPax{} is safe.
\begin{lemma}[Lemma~\ref{lem:composition}: Composition Lemma]
	 If some correct process decides a value $v$ in \fastbyz{} before an abort, then $v$ is the only value that can be decided in \prefPax{} with priorities as defined in Definition~\ref{def:priorities}. 
\end{lemma}

\begin{proof}
	To prove this lemma, we mainly rely on two properties: the \fastbyz{} Abort Agreement (Lemma~\ref{lem:abortAgreement}) and \prefPax{} Priority Decision (Lemma~\ref{lem:pref}).
	We consider two cases.
	
	\textit{Case 1.} Some correct follower process $p \neq p_1$ decided $v$ in \fastbyz. Then note that by Lemma~\ref{lem:abortAgreement}, all correct processes aborted with value $v$ and a correct unanimity proof. Since $n \geq 2f+1$, there are at least $f+1$ correct processes. Note that by the way we assign priorities to inputs of \prefPax{} in the composition of the two algorithms, all correct processes have inputs with the highest priority. Therefore, by Lemma~\ref{lem:pref}, the only decision value possible in \prefPax{} is $v$.
	Furthermore, note that by Lemma~\ref{lem:decisionAgreement}, if any other correct process decided in \fastbyz{}, that process's decision value was also $v$.
	
	\textit{Case 2.} Only the leader, $p_1$, decides in \fastbyz{}, and $p_1$ is correct. Then by Lemma~\ref{lem:abortAgreement}, all correct processes aborted with value $v$. Since $p_1$ is correct, $v$ is signed by $p_1$. It is possible that some of the processes also had a correct unanimity proof as their abort proof. However, note that in this scenario, all correct processes (at least $f+1$ processes) had inputs with either the highest or second highest priorities, all with the same abort value. Therefore, by Lemma~\ref{lem:pref}, the decision value must have been the value of one of these inputs. Since all these inputs had the same value $v$, $v$ must be the decision value of \prefPax.
\end{proof}

\begin{theorem}[End-to-end Validity]
	In the \composedAlgo{} algorithm, if there are no faulty processes and some process decides v, then v is the input of some process.
\end{theorem}

\begin{proof}
	Note that by Lemmas~\ref{lem:fastValidity} and~\ref{lem:pref}, this holds for each of the algorithms individually.
	Furthermore, recall that the abort values of \fastbyz{} become the input values of \prefPax{}, and the set-up phase does not invent new values. Therefore, we just have to show that if \fastbyz{} aborts, then all abort values are inputs of some process.
	Note that by the code in panic mode, if \fastbyz{} aborts, a process $p$ can output an abort value from one of three sources: its own Value region, the leader's Value region, or its own input value. Clearly, if its abort value is its input, then we are done. Furthermore note that a correct leader only writes its input in the Value region, and correct followers only write a copy of the leader's Value region in their own region. Since there are no faults, this means that only the input of the leader may be written in any Value region, and therefore all processes always abort with some processes input as their abort value.
\end{proof}

\begin{theorem}[End-to-end Agreement]
	In the \composedAlgo{} algorithm, if $p$ and $q$ are correct processes such that $p$ decides $v_1$ and $q$ decides $v_2$, then $v_1 = v_2$.
\end{theorem}

\begin{proof}
	First note that by Lemmas~\ref{lem:decisionAgreement} and~\ref{lem:pref}, each of the algorithms satisfy this individually. Thus Lemma~\ref{lem:composition} implies the theorem.
\end{proof}

\begin{theorem}[End-to-end Termination]\label{thm:e2eTerm}
	In \composedAlgo{} algorithm, if some correct process is eventually the sole leader forever, then every correct process eventually decides. 
\end{theorem}

\begin{proof}
    Assume towards a contradiction that some correct process $p$ is eventually the sole leader forever, and let $t$ be the time when $p$ last becomes leader. Now consider some process $q$ that has not decided before $t$. We consider several cases:
\begin{enumerate}
    \item If $q$ is executing \prefPax{} at time $t$, then $q$ will eventually decide, by termination of \prefPax{} (Lemma~\ref{lem:pref}).
    \item If $q$ is executing \fastbyz{} at time $t$, we distinguish two sub-cases:
    \begin{enumerate}
        \item $p$ is also executing as the leader of \fastbyz{} at time $t$. Then $p$ will eventually propose a value, so $q$ will either decide in \fastbyz{} or abort from \fastbyz{} (by Lemma~\ref{lem:fastTermination}) and decide in \prefPax{} by Lemma~\ref{lem:pref}.
        \item $p$ is executing in \prefPax{}. Then $p$ must have panicked and aborted from \fastbyz{}. Thus, $q$ will also abort from \fastbyz{} and decide in \prefPax{} by Lemma~\ref{lem:pref}.\qedhere
    \end{enumerate}
\end{enumerate}
\end{proof}

Note that to strengthen~\ref{thm:e2eTerm} to general termination as stated in our model, we require the additional standard assumption~\cite{lamport1998part} that some correct process $p$ is eventually the sole leader forever. In practice, however, $p$ does not need to be the sole leader forever, but rather \textit{long enough} so that all correct processes decide.


	\section{Correctness of Protected Memory Paxos}\label{sec:crashCorrectness}

In this section, we present the proof of Theorem~\ref{thm:crashAlgo}. We do so by showing the Algorithm~\ref{alg:consensus} is an algorithm that satisfies all of the properties in the theorem.

We first show that Algorithm~\ref{alg:consensus} correctly implements consensus, starting with validity. Intuitively, validity is preserved because each process that writes any value in a slot either writes its own value, or adopts a value that was previously written in a slot. We show that every value written in any slot must have been the input of some process.

\begin{theorem}[Validity]\label{thm:validity}
	In Algorithm~\ref{alg:consensus}, if a process $p$ decides a value $v$, then $v$ was the input to some process.
\end{theorem}

\begin{proof}
	Assume by contradiction that some process $p$ decides a value $v$ and $v$ is not the input of any process. Since $v$ is not the input value of $p$, then $p$ must have adopted $v$ by reading it from some process $p'$ at line~\ref{line:phase1read}. Note also that a process cannot adopt the initial value $\bot$, and thus, $v$ must have been written in $p'$'s memory by some other process. Thus we can define a sequence of processes $s_1, s_2,\ldots, s_k$, where $s_i$ adopts $v$ from the location where it was written by $s_{i+1}$ and $s_1=p$. This sequence is necessarily finite since there have been a finite number of steps taken up to the point when $p$ decided $v$. Therefore, there must be a last element of the sequence, $s_k$ who wrote $v$ in line~\ref{line:phase2write} without having adopted $v$. This implies $v$ was $s_k$'s input value, a contradiction. 
\end{proof}

We now focus on agreement. 

\begin{theorem}[Agreement]
	\label{thm:agreement}
	In Algorithm~\ref{alg:consensus}, for any processes $p$ and $q$, if $p$ and $q$ decide values $v_p$ and $v_q$ respectively, then $v_p = v_q$.
\end{theorem}

Before showing the proof of the theorem, we first introduce the following useful lemmas.

 \begin{lemma}\label{lem:atomic}
 	The values a leader accesses on remote memory cannot change between when it reads them and when it writes them.
 \end{lemma}

 \begin{proof}
 	Recall that each \disk{} only allows write-access to the most recent process that acquired it. In particular, that means that each \disk{} only gives access to one process at a time. Note that the only place at which a process acquires write-permissions on a \disk{} is at the very beginning of its run, before reading the values written on the \disk. In particular, for each \disk{} $d$ a process does not issue a read on $d$ before its permission request on $d$ successfully completes. Therefore, if a process $p$ succeeds in writing on \disk{} $m$, then no other process could have acquired $d$'s write-permission after $p$ did, and therefore, no other process could have changed the values written on $m$ after $p$'s read of $m$. 
\end{proof}

\begin{lemma}\label{lem:onevalue}
    If a leader writes values $v_i$ and $v_j$ at line~\ref{line:phase2write} with the same proposal number to memories $i$ and $j$, respectively, then $v_i = v_j$.
\end{lemma}

\begin{proof}
    Assume by contradiction that a leader $p$ writes different values $v_1 \neq v_2$ with the same proposal number. Since each thread of $p$ executes the phase 2 write (line~\ref{line:phase2write}) at most once per proposal number, it must be that different threads $T_1$ and $T_2$ of $p$ wrote $v_1$ and $v_2$, respectively. If $p$ does not perform phase 1 (i.e., if $p = p_1$ and this is $p$'s first attempt), then it is impossible for $T_1$ and $T_2$ to write different values, since CurrentVal was set to $v$ at line~\ref{line:currentValInit} and was not changed afterwards. Otherwise (if $p$ performs phase 1), then let $t_1$ and $t_2$ be the times when $T_1$ and $T_2$ executed line~\ref{line:phase2start}, respectively ($T_1$ and $T_2$ must have done so, since we assume that they both reached the phase 2 write at line~\ref{line:phase2write}). Assume \textit{wlog} that $t_1 \leq t_2$. Due to the check and abort at line ~\ref{line:restartStraggler}, CurrentVal cannot change after $t_1$ while keeping the same proposal number. Thus, when $T_1$ and $T_2$ perform their phase 2 writes (after $t_1$), CurrentVal has the same value as it did at $t_1$; it is therefore impossible for $T_1$ and $T_2$ to write different values. We have reached a contradiction.
\end{proof}

\begin{lemma}
    If a process $p$ performs phase 1 and then writes to some memory $m$ with proposal number $b$ at line~\ref{line:phase2write}, then $p$ must have written $b$ to $m$ at line~\ref{line:phase1write} and read from $m$ at line~\ref{line:phase1read}.
\end{lemma}

\begin{proof}
    Let $T$ be the thread of $p$ which writes to $m$ at line~\ref{line:phase2write}. If phase 1 is performed (i.e., the condition at line~\ref{line:checkOmitPhase1} is satisfied), then by construction $T$ cannot reach line~\ref{line:phase2write} without first performing lines~\ref{line:phase1write} and~\ref{line:phase1read}. Since $T$ only communicates with $m$, it must be that lines~\ref{line:phase1write} and~\ref{line:phase1read} are performed on $m$.
\end{proof}

\begin{proof}[Proof of Theorem~\ref{thm:agreement}]
Assume by contradiction that $v_p \neq v_q$. Let $b_p$ and $b_q$ be the proposal numbers with which $v_p$ and $v_q$ are decided, respectively. Let $W_p$ (resp. $W_q$) be the set of \disks{} to which $p$ (resp. $q$) successfully wrote in phase 2 line~\ref{line:phase2write} before deciding $v_p$ (resp. $v_q$). Since $W_p$ and $W_q$ are both majorities, their intersection must be non-empty. Let $m$ be any \disk{} in $W_p \cap W_q$. 

We first consider the case in which one of the processes did not perform phase 1 before deciding (i.e., one of the processes is $p_1$ and it decided on its first attempt). Let that process be $p$ \textit{wlog}. Further assume \textit{wlog} that $q$ is the first process to enter phase 2 with a value different from $v_p$. $p$'s phase 2 write on $m$ must have preceded $q$ obtaining permissions from $m$ (otherwise, $p$'s write would have failed due to lack of permissions). Thus, $q$ must have seen $p$'s value during its read on $m$ at line~\ref{line:phase1read}, and thus $q$ cannot have adopted its own value. Since $q$ is the first process to enter phase 2 with a value different from $v_p$, $q$ cannot have adopted any other value than $v_p$, so $q$ must have adopted $v_p$. Contradiction.

We now consider the remaining case: both $p$ and $q$ performed phase 1 before deciding. We assume \textit{wlog} that $b_p < b_q$ and that $b_q$ is the smallest proposal number larger than $b_p$ for which some process enters phase 2 with CurrentVal $\neq v_p$.


Since $b_p < b_q$, $p$'s read at $m$ must have preceded $q$'s phase 1 write at $m$ (otherwise $p$ would have seen $q$'s larger proposal number and aborted). This implies that $p$'s phase 2 write at $m$ must have preceded $q$'s phase 1 write at $m$ (by Lemma~\ref{lem:atomic}). Thus $q$ must have seen $v_p$ during its read and cannot have adopted its own input value. However, $q$ cannot have adopted $v_p$, so $q$ must have adopted $v_q$ from some other slot $sl$ that $q$ saw during its read. It must be that $sl.minProposal < b_q$, otherwise $q$ would have aborted. Since $sl.minProposal \geq sl.accProposal$ for any slot, it follows that $sl.accProposal < b_q$. If $sl.accProposal < b_p$, $q$ cannot have adopted $sl.value$ in line~\ref{line:adoptValue} (it would have adopted $v_p$ instead). Thus it must be that  $b_p \leq sl.accProposal < b_q$; however, this is impossible because we assumed that $b_q$ is the smallest proposal number larger than $b_p$ for which some process enters phase 2 with CurrentVal $\neq v_p$. We have reached a contradiction.
%
\end{proof}

Finally, we prove that the termination property holds.
\begin{theorem}[Termination]\label{thm:termination}
	Eventually, all correct processes decide. 
\end{theorem}

\begin{lemma}\label{lem:weakTermination}
    If a correct process $p$ is executing the for loop at lines~\ref{line:pforStart}--\ref{line:pforEnd}, then $p$ will eventually exit from the loop.
\end{lemma}

\begin{proof}
    The threads of the for loop perform the following potentially blocking steps: obtaining permission (line~\ref{line:permission}), writing (lines~\ref{line:phase1write} and \ref{line:phase2write}), reading (line~\ref{line:phase1read}), and waiting for other threads (the barrier at line~\ref{line:barrierPhase2} and the exit condition at line~\ref{line:pforEnd}). By our assumption that a majority of memories are correct, a majority of the threads of the for loop must eventually obtain permission in line~\ref{line:permission} and invoke the write at line~\ref{line:phase1write}. If one of these writes fails due to lack of permission, the loop aborts and we are done. Otherwise, a majority of threads will perform the read at line~\ref{line:phase1read}. If some thread aborts at lines~\ref{line:restart2} and \ref{line:Crestart1}, then the loop aborts and we are done. Otherwise, a majority of threads must add themselves to ListOfReady, pass the barrier at line~\ref{line:barrierPhase2} and invoke the write at line~\ref{line:phase2write}. If some such write fails, the loop aborts; otherwise, a majority of threads will reach the check at line~\ref{line:pforEnd} and thus the loop will exit.
\end{proof}

\begin{proof}[Proof of Theorem~\ref{thm:termination}]
    The $\Omega$ failure detector guarantees that eventually, all processes trust the same correct process $p$. Let $t$ be the time after which all correct processes trust $p$ forever. By Lemma~\ref{lem:weakTermination}, at some time $t'\geq t$, all correct processes except $p$ will be blocked at line~\ref{line:leaderWait}. Therefore, the \textit{minProposal} values of all \disks{}, on all slots except those of $p$ stop increasing. Thus, eventually, $p$ picks a \textit{propNr} that is larger than all others written on any memory, and stops restarting at line~\ref{line:Crestart1}. Furthermore, since no process other than $p$ is executing any steps of the algorithm, and in particular, no process other than $p$ ever acquires any \disk{} after time $t'$, $p$ never loses its permission on any of the \disks. So, all writes executed by $p$ on any correct \disk{} must return $ack$. Therefore, $p$ will decide and broadcast its decision to all. All correct processes will receive $p$'s decision and decide as well.
\end{proof}

To complete the proof of Theorem~\ref{thm:crashAlgo}, we now show that Algorithm~\ref{alg:consensus} is 2-deciding.

\begin{theorem}
	Algorithm~\ref{alg:consensus} is 2-deciding.
\end{theorem}

\begin{proof}
	Consider an execution in which $p_1$ is timely, and no process's failure detector ever suspects $p_1$. Then, since no process thinks itself the leader, and processes do not deviate from their protocols, no process calls changePermission on any \disk{}. Furthermore, $p_1$'s does not perform phase 1 (lines~\ref{line:checkOmitPhase1}--\ref{line:endPhase1}), since it is $p_1$'s first attempt. Thus, since $p_1$ initially has write permission on all \disks{}, all of $p_1$'s phase 2 writes succeed. Therefore, $p_1$ terminates, deciding its own proposed value $v$, after one write per \disk.
\end{proof}

	\section{Pseudocode for the \combinedAlgo}\label{sec:alignedPseudo}

\renewcommand{\figurename}{Algorithm}

\begin{figure}[H]
	\caption{\combinedAlgo}
	\begin{lstlisting}[columns=fullflexible,breaklines=true, keywords={}]
A=(P, M) is set of acceptors
propose(v){
	resps = [] //prepare empty responses list
	choose propNr bigger than any seen before
	for all a in A{
		cflag = @\textbf{communicate1}@(a, propNr)
		resp = @\textbf{hearback1}@(a)
		if (cflag){resps.append((a, resp))	}	}
	wait until resps has responses from a majority of A
	next = @\textbf{analyze1}@(resps)
	if (next == RESTART) restart;
	resps = []
	for all a in A{
		cflag = @\textbf{communicate2}@(a, next)
		resp = @\textbf{hearback2}@(a)
		if (cflag){resps.append((a, resp))	}	}
	wait until resps has responses from a majority of A
	next = @\textbf{analyze2}@(resps)
	if (next == RESTART) restart;
	decide next; }
	\end{lstlisting}
	\label{alg:combined}
\end{figure}

\begin{figure}[h!]
	\caption{Communicate Phase 1---code for process $p$}
	\begin{lstlisting}[columns=fullflexible,breaklines=true, keywords={}]
bool communicate1(agent a, value propNr){
	if (a is @\disk@){
		changePermission(a, {(R:@$\Pi$@-{p}, W:@$\emptyset$@, RW: {p}); //acquire write permission
		return write(a[p], {propNr, -, -}) }
	else{
		send prepare(propNr) to a	
		return true	}	}
	\end{lstlisting}
	\label{alg:communicate1}
\end{figure}

\begin{figure}[h!]
	\caption{Hear Back Phase 1---code for process $p$}
	\begin{lstlisting}[columns=fullflexible,breaklines=true, keywords={}]
value hearback1(agent a){
	if (a is @\disk@){
		for all processes q{
			localInfo[q] = read(a[q])	}
		return = localInfo}
	else{
		return value received from a	}	}
	\end{lstlisting}
	\label{alg:hear1}
\end{figure}

\begin{figure}[H]
	\caption{Analyze Phase 1---code for process $p$}
	\begin{lstlisting}[columns=fullflexible,breaklines=true, keywords={}]
responses is a list of (agent, response) pairs
value analyze1(responses resps){
	for resp in resps {
		if (resp.agent is @\disk@){
			for all slots s in v.info {
				if (s.minProposal @$>$@ propNr) return RESTART	}	
			(v, accProposal) = value and accProposal of slot with highest 
			accProposal that had a value	}	} 
	return v where (v, accProposal) is the highest accProposal seen in resps.response of all agents }
	\end{lstlisting}
	\label{alg:analyze1}
\end{figure}

\begin{figure}[h!]
	\caption{Communicate Phase 2---code for process $p$}
	\begin{lstlisting}[columns=fullflexible,breaklines=true, keywords={}]
bool communicate2(agent a, value msg){
	if (a is @\disk@){
		return write(a[p], {msg.propNr, msg.propNr, msg.val})}
	else{
		send accepted(msg.propNr, msg.val) to a	
		return true	}	}
	\end{lstlisting}
	\label{alg:communicate2}
\end{figure}

\begin{figure}[h!]
	\caption{Hear Back Phase 2---code for process $p$}
	\begin{lstlisting}[columns=fullflexible,breaklines=true, keywords={}]
value hearback2(agent a){
	if (a is @\disk@){
		return ack	}
	else{
		return value received from a	}	}
	\end{lstlisting}
	\label{alg:hear2}
\end{figure}

\begin{figure}[h!]
	\caption{Analyze Phase 2---code for process $p$}
	\begin{lstlisting}[columns=fullflexible,breaklines=true, keywords={}]
value analyze2(value v, responses resps){
	if there are at least A/2 + 1 resps such that resp.response==ack{
		return v }
	return RESTART	}
	\end{lstlisting}
	\label{alg:analyze2}
\end{figure}

\renewcommand{\figurename}{Figure}

	\fi
\end{document}